\RequirePackage{xcolor}
\documentclass[journal,twoside,web]{ieeecolor}
\PassOptionsToPackage{usenames,dvipsnames}{xcolor}
\usepackage{generic}
\usepackage{cite}
\usepackage{amsmath,amssymb,amsfonts}
\usepackage{algorithmic}
\usepackage{graphicx}
\usepackage{algorithm,algorithmic}
\usepackage{hyperref}
\hypersetup{hidelinks=true}
\usepackage{textcomp}
\usepackage{graphics} 
\usepackage{epsfig} 
\usepackage{mathptmx} 
\usepackage{times} 
\usepackage[acronym]{glossaries}
\usepackage{cleveref}
\usepackage{graphicx}
\usepackage{subcaption}
\usepackage{pstricks} 				
\usepackage{pgfplots} 				
\usepackage{pgfplotstable} 			
\usepackage{tikz}
\usepackage{nicematrix}
\usepackage{commath}
\usepackage{algorithm}
\usepackage{mathtools}
\usepackage{thmtools}
\usepackage{afterpage}
\usepackage{siunitx}
\usepackage{scalerel}
\usepackage{relsize, exscale}
\usepackage{nccmath}
\def\BibTeX{{\rm B\kern-.05em{\sc i\kern-.025em b}\kern-.08em
    T\kern-.1667em\lower.7ex\hbox{E}\kern-.125emX}}
\newacronym{LTI}{LTI}{linear time-invariant}
\newacronym{SDP}{SDP}{semidefinite program}
\newacronym{MIP}{MIP}{mixed-integer program}
\newacronym{IQC}{IQC}{integral quadratic constraint}
\newacronym{SOS}{SOS}{sum of squares}
\newacronym{ROA}{RoA}{region of attraction}
\newacronym{ReLU}{ReLU}{rectified linear Unit}
\newacronym{IBP}{IBP}{interval bound propagation}
\newacronym{REN}{REN}{recurrent equilibrium network}
\newacronym{RNN}{RNN}{recurrent neural network}
\newacronym{CNN}{CNN}{convolutional neural network}
\newacronym{GAS}{GAS}{globally asymptotically stable}
\newacronym{LAS}{LAS}{locally asymptotically stable}
\newacronym[longplural=Linear Matrix Inequalities]{LMI}{LMI}{linear matrix inequality}
\newacronym{LSTM}{LSTM}{long short-term memory}
\newacronym{LQR}{LQR}{linear-quadratic regulator}
\newacronym{MPC}{MPC}{model predictive controller}
\newacronym{DoF}{DoF}{degree of freedom}
\newacronym{RHP}{RHP}{right-half plane}
\newacronym{EoM}{EoM}{equations of motion}
\newacronym{CoG}{CoG}{center of gravity}
\newacronym{PRBS}{PRBS}{pseudorandom-binary signal}
\newacronym{MSE}{MSE}{mean-squared error}
\newacronym{NNC}{NNC}{neural-network-based controller}
\newacronym{RMSE}{RMSE}{root mean square error}
\newacronym{MAE}{MAE}{maximum absolute error}
\newcommand{\realsN}[1]{\ensuremath{\mathbb{R}^{#1}}}
\newcommand{\integersN}[1]{\ensuremath{\mathbb{Z}^{#1}}}

\newcommand{\transpose}{^\top}

\DeclareMathOperator*{\argmin}{arg\,min}

\DeclareMathOperator{\graph}{gr}
\DeclareMathAlphabet{\mathcal}{OMS}{cmsy}{m}{n}
\useRomanappendicesfalse
\usepgfplotslibrary{groupplots} 	
\usepgfplotslibrary{colormaps}
\usetikzlibrary{plotmarks}
\usetikzlibrary{calc}
\usetikzlibrary{shapes.geometric}
\usetikzlibrary{arrows.meta,arrows}
\usetikzlibrary{patterns,patterns.meta}
\usetikzlibrary{decorations.pathreplacing,calligraphy}
\pgfplotsset{compat=newest}
\pgfplotsset{plot coordinates/math parser=false}
\newlength\figureheight
\newlength\figurewidth
\pgfplotsset{/pgfplots/layers/niceLayers/.define layer set={
		axis background,axis grid,main,axis ticks,axis lines,axis tick labels,axis descriptions,axis foreground
	}{/pgfplots/layers/standard}
}
\pgfplotsset{every axis/.append style={
		set layers=niceLayers,
		tick label style={font=\scriptsize},
		clip marker paths=true,
		line width=1.5pt,
		line cap=round,
		line join=round,
		tick style={semithick, color=black}
}}
\NiceMatrixOptions{xdots/shorten=0.8em}
\NiceMatrixOptions{cell-space-limits = 1pt}
\newtheorem{thm}{Theorem}[section]

\newtheorem{cor}[thm]{Corollary}
\newtheorem{defn}[thm]{Definition}

\newtheorem{exmp}{Example}
\newtheorem{rem}{Remark}
\newtheorem{assume}{Assumption}
\crefname{thm}{Theorem}{Theorems}
\crefname{lemma}{Lemma}{Lemmas}
\crefname{prop}{Proposition}{Propositions}
\crefname{cor}{Corollary}{Corollaries}
\crefname{defn}{Definition}{Definitions}
\crefname{conj}{Conjecture}{Conjectures}
\crefname{exmp}{Example}{Examples}
\crefname{rem}{Remark}{Remarks}
\crefname{assume}{Assumption}{Assumptions}
\crefname{equation}{}{} 
\Crefname{equation}{Equation}{Equations} 
\crefname{figure}{Fig.}{Figs.} 
\Crefname{figure}{Fig.}{Figs.} 
\crefname{subfigure}{Fig.}{Figs.} 
\Crefname{subfigure}{Fig.}{Figs.} 
\crefname{section}{Section}{Sections} 
\crefname{algorithm}{Algorithm}{Algorithms} 
\Crefname{algorithm}{Algorithm}{Algorithms} 
\DeclareSIUnit{\deg}{deg}
\usepackage{fancyhdr}
\fancypagestyle{firstpage}{
  \fancyhf{}
  
  \fancyhead[C]{\footnotesize This work has been submitted to the IEEE for possible publication. Copyright may be transferred without notice, after which this version may no longer be accessible.}
}
\begin{document}
\title{Contributions to Semialgebraic-Set-Based Stability Verification of Dynamical Systems with Neural-Network-Based Controllers}
\author{Alvaro Detailleur, Dalim Wahby, Guillaume Ducard, Christopher Onder
\thanks{A. Detailleur and C. Onder are with the Institute for Dynamic Systems and Control (IDSC), Department of Mechanical and Process Engineering, ETH Zurich, Leonhardstrasse 21, 8092 Zürich, Switzerland. (E-mail: adetailleur@student.ethz.ch; onder@idsc.mavt.ethz.ch, Telephone: +41 44 632 87 96)}%
\thanks{D. Wahby and G. Ducard are with Universit{\'e} C\^{o}te d`Azur, i3S, CNRS, 06903 Sophia Antipolis, France. (E-mail: dalim.wahby@univ-cotedazur.fr; guillaume.ducard@univ-cotedazur.fr). This work was supported by the French government through the France 2030 investment plan managed by the National Research Agency (ANR), as part of the Initiative of Excellence Universit{\'e} C\^{o}te d`Azur under reference number ANR- 15-IDEX-01.}%
}

\maketitle
\thispagestyle{firstpage}
\begin{abstract}
\Acrfullpl{NNC} can represent complex, highly nonlinear control laws, but verifying the closed-loop stability of dynamical systems using them remains challenging. This work presents contributions to a state-of-the-art stability verification procedure for \acrshort{NNC}-controlled systems which relies on semialgebraic-set-based input-output modeling to pose the search for a Lyapunov function as an optimization problem.
Specifically, this procedure’s conservatism when analyzing \acrshortpl{NNC} using transcendental activation functions and the restriction to feedforward \acrshortpl{NNC} are addressed by a) introducing novel semialgebraic activation functions that preserve key properties of common transcendental activations and b) proving compatibility of \acrshortpl{NNC} from the broader class of \acrfullpl{REN} with this procedure. Furthermore, the indirect optimization of a local \acrfull{ROA} estimate using a restricted set of candidate Lyapunov functions is greatly improved via c) the introduction of a richer parameterization of candidate Lyapunov functions than previously reported and d) the formulation of novel \acrfullpl{SDP} that directly optimize the resulting \acrshort{ROA} estimate. The value of these contributions is highlighted in two numerical examples.
\end{abstract}

\begin{IEEEkeywords}
closed-loop stability, neural networks, semidefinite programming (SDP), sum of squares (SOS)
\end{IEEEkeywords}

\section{Introduction}
\label{sec:Introduction}
\IEEEPARstart{T}{he} interest of the control engineering community in (artificial) neural networks dates back to at least the end of the 1980s \cite{mybibfile:Hunt1992}. Researchers have long recognized that neural networks possess several properties that make them uniquely suited for use in control applications. In particular:
\begin{enumerate}
    \item According to (variants of) the universal function approximation theorem \cite{mybibfile:Hanin2019}, large enough neural networks are capable of approximating continuous functions with arbitrary accuracy, which implies neural networks can be used to parameterize highly complex control laws.
    \item Neural networks represent a method to store such complex parameterizations whilst requiring relatively low computational requirements compared to traditional optimization algorithms used in (optimal) control, allowing them to be evaluated on embedded hardware platforms and/or at a high frequency \cite{mybibfile:Gonzalez2024, mybibfile:Detailleur2025}.
\end{enumerate}

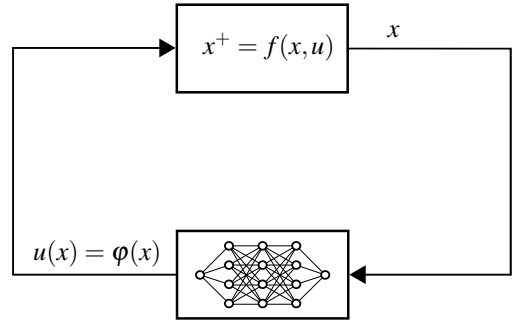
\begin{figure}[t]
	\centering
	\def\svgwidth{0.75\columnwidth}
\begingroup%
  \makeatletter%
  \providecommand\color[2][]{%
    \errmessage{(Inkscape) Color is used for the text in Inkscape, but the package 'color.sty' is not loaded}%
    \renewcommand\color[2][]{}%
  }%
  \providecommand\transparent[1]{%
    \errmessage{(Inkscape) Transparency is used (non-zero) for the text in Inkscape, but the package 'transparent.sty' is not loaded}%
    \renewcommand\transparent[1]{}%
  }%
  \providecommand\rotatebox[2]{#2}%
  \newcommand*\fsize{\dimexpr\f@size pt\relax}%
  \newcommand*\lineheight[1]{\fontsize{\fsize}{#1\fsize}\selectfont}%
  \ifx\svgwidth\undefined%
    \setlength{\unitlength}{195.70528903bp}%
    \ifx\svgscale\undefined%
      \relax%
    \else%
      \setlength{\unitlength}{\unitlength * \real{\svgscale}}%
    \fi%
  \else%
    \setlength{\unitlength}{\svgwidth}%
  \fi%
  \global\let\svgwidth\undefined%
  \global\let\svgscale\undefined%
  \makeatother%
  \begin{picture}(1,0.63338912)%
    \lineheight{1}%
    \setlength\tabcolsep{0pt}%
    \put(0,0){\includegraphics[width=\unitlength,page=1]{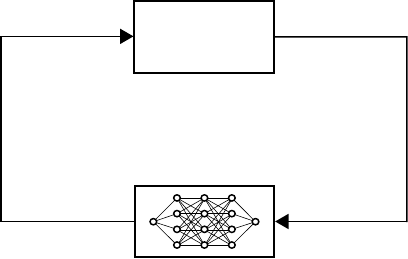}}%
    \put(0.3824088,0.53274254){\makebox(0,0)[lt]{\lineheight{1.25}\smash{\begin{tabular}[t]{l}$x^+ = f(x, u)$\end{tabular}}}}%
    \put(0.04233083,0.12024308){\makebox(0,0)[lt]{\lineheight{1.25}\smash{\begin{tabular}[t]{l}$u(x) = \varphi(x)$\end{tabular}}}}%
    \put(0.75198126,0.56578001){\makebox(0,0)[lt]{\lineheight{1.25}\smash{\begin{tabular}[t]{l}$x$\end{tabular}}}}%
  \end{picture}%
\endgroup%

	\caption{Block diagram of a general discrete-time closed-loop dynamical system with a neural-network-based controller $\varphi$. The stability properties of systems of this form are examined in this work.}
	\label{fig:GeneralClosedLoopSystem}
\end{figure}

However, establishing stability properties of closed-loop systems utilizing a neural network as a controller, as shown in \cref{fig:GeneralClosedLoopSystem}, is a non-trivial task \cite{mybibfile:Norris2021}. Traditional stability conditions used for \acrfull{LTI} systems are not directly applicable to certify stability for systems with so-called neural-network-based controllers \cite[Part~2]{mybibfile:Liberzon2002}. In addition, deciding the asymptotic stability of a closed-loop system containing a neural network made up of a single neuron is in general already an NP-Hard problem\cite{mybibfile:Blondel1999, mybibfile:Korda2022}. This has motivated research into stability verification procedures for this class of systems.

\subsection{Optimization-based Stability Verification Procedures}
Given the size and complexity of a general \acrshort{NNC}, many stability verification procedures pose the search for a stability certificate as an optimization program, e.g. a \acrfull{SDP} or \acrfull{MIP}. 
Previous successful approaches include \acrshort{MIP}-based methods that certify stability by comparing the \acrshort{NNC} to an existing robust optimization-based controller \cite{mybibfile:Dubach2022,mybibfile:Schwan2023}. Other methods have approximated forward reachable sets of the closed-loop system using \acrshortpl{SDP} \cite{mybibfile:Hu2020}. In this work we examine a stability verification procedure belonging to a third class of methods that pose the search for a Lyapunov function for the closed-loop system as one or more \acrshortpl{SDP} \cite{mybibfile:Korda2017,mybibfile:Korda2022,mybibfile:Richardson2023,mybibfile:Yin2022,mybibfile:Pauli2021,mybibfile:Newton2022,mybibfile:Revay2024,mybibfile:Samanipour2024}.
Procedures from this final category consist of a two-step process, executed sequentially:
\begin{enumerate}
    \item \textbf{Modeling Step}: Obtain a system model of the closed-loop system of \cref{fig:GeneralClosedLoopSystem}, which specifies how the states of the dynamical system change over time.
    \item \textbf{SDP Formulation Step}:
    Using the closed-loop system model and a parameterized set of (candidate Lyapunov) functions, formulate and solve one or more \acrshortpl{SDP} to find a Lyapunov function decreasing over time for all state trajectories defined admissible by the system model.
\end{enumerate}

Due to the nonlinear nature of \acrshortpl{NNC}, explicitly describing how the states of the closed-loop system change over time is generally infeasible. Therefore, various modeling techniques to obtain the required system model have been reported in the literature, including sector constraints \cite{mybibfile:Richardson2023,mybibfile:Yin2022}, semialgebraic sets \cite{mybibfile:Korda2017,mybibfile:Korda2022,mybibfile:Newton2022}, Zames-Falb multipliers \cite{mybibfile:Pauli2021}, (integral) quadratic constraints \cite{mybibfile:Yin2022} and incremental quadratic constraints \cite{mybibfile:Revay2024}.
Once a system model has been obtained using one (or more) of these techniques, it is used to formulate \acrshort{SDP}(s) searching for incremental, global or local Lyapunov functions for the closed-loop system within a parameterized set of functions. 

\subsection{Structure and Contributions}
This work focuses specifically on a stability verification procedure that uses a semialgebraic-set-based system model to pose the search for a global or local discrete-time Lyapunov function as an \acrshort{SDP} \cite{mybibfile:Korda2017,mybibfile:Korda2022}. The choice for this procedure is motivated by previous results \cite{mybibfile:Korda2022} demonstrating its ability to scale to large networks and exactly describe the input-output properties of a subset of \acrshortpl{NNC}, which stands in contrast to the inexact system models obtained using other modeling techniques. 

To motivate the contributions presented in this work, \cref{sec:SoAReview} first presents a review and identifies limitations of this state-of-the-art semialgebraic-set-based stability verification procedure. The contributions of this work address the limitations identified in this review and are presented in \cref{sec:ModelingContributions,sec:StabilityVerificationContributions}. 

\Cref{sec:ModelingContributions} details our contributions to the modeling step of the reviewed procedure. These expand the class of \acrshortpl{NNC} whose input-output properties can be described exactly using a semialgebraic-set-based model. Specifically, we
    \begin{enumerate}
        \item[a)] introduce two novel semialgebraic activation functions mimicking the fundamental properties of the common $\text{softplus}$ and $\text{tanh}$ activation functions, and
        \item[b)]prove compatibility of this modeling technique with \acrshortpl{NNC} from the class of \acrfullpl{REN}, which includes, among others, \acrfullpl{RNN}, \acrfull{LSTM} networks and \acrfullpl{CNN} \cite{mybibfile:Revay2024}.
    \end{enumerate}

\Cref{sec:StabilityVerificationContributions} presents our contributions to the \acrshort{SDP} formulation step of the reviewed procedure. They consist of
\begin{enumerate}
    \item[c)] the introduction of a richer class of candidate Lyapunov functions compatible with the existing local stability analysis procedure, particularly for networks containing ReLU or any of the semialgebraic activation functions newly presented in this work, and
    \item[d)] an improved local stability analysis procedure consisting of a sequence of \acrshortpl{SDP} directly optimizing an inner estimate of the closed-loop system's \acrshort{ROA}, 
\end{enumerate}
which represent improvements over current state-of-the-art procedures to examine local stability properties \cite{mybibfile:Yin2022,mybibfile:Pauli2021,mybibfile:Newton2022}.

Finally, \cref{sec:NumericalResults} presents two numerical examples to highlight the value of our contributions. These examples showcase the enlarged class of closed-loop systems suitable for analysis and the improved method for obtaining inner estimates of the closed-loop system's \acrshort{ROA}.
A conclusion and an outlook for future work are presented in \cref{sec:Conclusion}.

\subsection{Notation}
This work uses the following notational conventions:
\begin{itemize}
    \item In the analysis of discrete-time systems, the plus superscript indicates the successor variable, e.g. $x, \ x^+ \in \realsN{n_x}$ represent the current and successor state, respectively.
    \item All inequalities are defined element-wise. $P \succ 0$, $P \succeq 0$ denote a positive definite and positive semidefinite matrix $P$, respectively.
    \item The set $\{1, \dots, n\}$ is denoted as $[n]$. The notation $\mathcal{M}(x, \,n)$ denotes the vector of all unique products of $n$ entries of $x$.  
    $[A_{22}]_i$ refers to the $i$-th row of matrix $A_{22}$.  Entries denoted by $*$ in a matrix indicate the matrix is assumed symmetric.
    \item The identity map is denoted $\textrm{id}$. Function composition is denoted by $\circ$, i.e. $f \circ g(x) = f\big(g(x)\big)$.
    \item The interior of a set $\mathcal{Q}$ is denoted $\text{int}(\mathcal{Q})$.
    \item The Minkowski sum and Pontryagin difference are denoted by $\oplus$ and $\ominus$, respectively. Elementwise multiplication is denoted by $\odot$.
    \item Unless specified, all norms represent the Euclidean norm.
\end{itemize}

\section{Review of the state of the art}
\label{sec:SoAReview}

\begin{figure*}[t]
    \centering
        \begin{subfigure}{0.33\textwidth} 
        \centering
        \def\svgwidth{0.95\linewidth}
\begingroup%
  \makeatletter%
  \providecommand\color[2][]{%
    \errmessage{(Inkscape) Color is used for the text in Inkscape, but the package 'color.sty' is not loaded}%
    \renewcommand\color[2][]{}%
  }%
  \providecommand\transparent[1]{%
    \errmessage{(Inkscape) Transparency is used (non-zero) for the text in Inkscape, but the package 'transparent.sty' is not loaded}%
    \renewcommand\transparent[1]{}%
  }%
  \providecommand\rotatebox[2]{#2}%
  \newcommand*\fsize{\dimexpr\f@size pt\relax}%
  \newcommand*\lineheight[1]{\fontsize{\fsize}{#1\fsize}\selectfont}%
  \ifx\svgwidth\undefined%
    \setlength{\unitlength}{271.41416893bp}%
    \ifx\svgscale\undefined%
      \relax%
    \else%
      \setlength{\unitlength}{\unitlength * \real{\svgscale}}%
    \fi%
  \else%
    \setlength{\unitlength}{\svgwidth}%
  \fi%
  \global\let\svgwidth\undefined%
  \global\let\svgscale\undefined%
  \makeatother%
  \begin{picture}(1,0.30529543)%
    \lineheight{1}%
    \setlength\tabcolsep{0pt}%
    \put(0,0){\includegraphics[width=\unitlength,page=1]{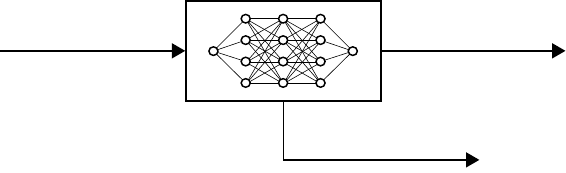}}%
    \put(0.59812475,0.05068204){\makebox(0,0)[lt]{\lineheight{1.25}\smash{\begin{tabular}[t]{l}$\lambda(x)$\end{tabular}}}}%
    \put(0.11245526,0.24259885){\makebox(0,0)[lt]{\lineheight{1.25}\smash{\begin{tabular}[t]{l}$x$\end{tabular}}}}%
    \put(0.81139142,0.2426014){\makebox(0,0)[lt]{\lineheight{1.25}\smash{\begin{tabular}[t]{l}$\varphi(x)$\end{tabular}}}}%
  \end{picture}%
\endgroup%

        \caption{}
        \label{subfig:NeuralNetworkGraph}
    \end{subfigure}
    \hfill
    \begin{subfigure}{0.65\textwidth}
        \centering
        \def\svgwidth{0.95\linewidth}
\begingroup%
  \makeatletter%
  \providecommand\color[2][]{%
    \errmessage{(Inkscape) Color is used for the text in Inkscape, but the package 'color.sty' is not loaded}%
    \renewcommand\color[2][]{}%
  }%
  \providecommand\transparent[1]{%
    \errmessage{(Inkscape) Transparency is used (non-zero) for the text in Inkscape, but the package 'transparent.sty' is not loaded}%
    \renewcommand\transparent[1]{}%
  }%
  \providecommand\rotatebox[2]{#2}%
  \newcommand*\fsize{\dimexpr\f@size pt\relax}%
  \newcommand*\lineheight[1]{\fontsize{\fsize}{#1\fsize}\selectfont}%
  \ifx\svgwidth\undefined%
    \setlength{\unitlength}{511.75989743bp}%
    \ifx\svgscale\undefined%
      \relax%
    \else%
      \setlength{\unitlength}{\unitlength * \real{\svgscale}}%
    \fi%
  \else%
    \setlength{\unitlength}{\svgwidth}%
  \fi%
  \global\let\svgwidth\undefined%
  \global\let\svgscale\undefined%
  \makeatother%
  \begin{picture}(1,0.15725041)%
    \lineheight{1}%
    \setlength\tabcolsep{0pt}%
    \put(0,0){\includegraphics[width=\unitlength,page=1]{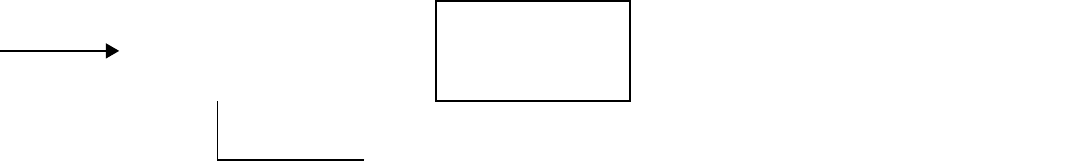}}%
    \put(0.42310011,0.10431159){\makebox(0,0)[lt]{\lineheight{1.25}\smash{\begin{tabular}[t]{l}$x^+ = f(x,u)$\end{tabular}}}}%
    \put(0,0){\includegraphics[width=\unitlength,page=2]{OpenLoop_SemialgebraicSet.pdf}}%
    \put(0.24653124,0.02215671){\makebox(0,0)[lt]{\lineheight{1.25}\smash{\begin{tabular}[t]{l}$\lambda(x)$\end{tabular}}}}%
    \put(0.04243475,0.12394066){\makebox(0,0)[lt]{\lineheight{1.25}\smash{\begin{tabular}[t]{l}$x$\end{tabular}}}}%
    \put(0.32816359,0.12394201){\makebox(0,0)[lt]{\lineheight{1.25}\smash{\begin{tabular}[t]{l}$\varphi(x)$\end{tabular}}}}%
    \put(0,0){\includegraphics[width=\unitlength,page=3]{OpenLoop_SemialgebraicSet.pdf}}%
    \put(0.83862644,0.02215671){\makebox(0,0)[lt]{\lineheight{1.25}\smash{\begin{tabular}[t]{l}$\lambda^+(x)$\end{tabular}}}}%
    \put(0.61482338,0.12394066){\makebox(0,0)[lt]{\lineheight{1.25}\smash{\begin{tabular}[t]{l}$x^+(x)$\end{tabular}}}}%
    \put(0.91078868,0.12394201){\makebox(0,0)[lt]{\lineheight{1.25}\smash{\begin{tabular}[t]{l}$\varphi^+(x)$\end{tabular}}}}%
  \end{picture}%
\endgroup%

        \caption{}
        \label{subfig:OpenLoopGraph}
    \end{subfigure}
    \caption{Schematic of \subref{subfig:NeuralNetworkGraph} the neural network $\varphi$, which is modeled via its graph by the semialegbraic set $\mathbf{K}_\varphi$ and \subref{subfig:OpenLoopGraph} the composed loop $L = \varphi \circ f \circ (\textrm{id}, \varphi)$, which is modeled via its graph by the semialegbraic set $\mathbf{K}_{L}$.}
    \label{fig:SemialgebraicSetModel}
\end{figure*}

This work examines the stability properties of an open-loop system with polynomial dynamics, $f(x,u)\colon \realsN{n_x} \times \realsN{n_u} \mapsto \realsN{n_x}$, controlled by a fixed, memoryless \acrshort{NNC}, $u(x) = \varphi(x)\colon \realsN{n_x} \mapsto \realsN{n_u}$, as shown in \cref{fig:GeneralClosedLoopSystem}. This discrete-time closed-loop system 
\begin{equation}
    \label{eq:GeneralClosedLoopSystem}
    x^+ = f\big(x, \varphi(x)\big),
\end{equation}
is assumed to be locally bounded and have an equilibrium point at $x = 0$, i.e. $f\big(0, \varphi(0)\big) = 0$. As the nonlinear mapping defined by a general \acrshort{NNC} prevents an explicit description of the closed-loop system's dynamics from being obtained, 
systematic methods capable of finding a certificate of the following stability properties are of great interest.
\begin{defn}[Global Asymptotic Stability]
    Closed-loop system \cref{eq:GeneralClosedLoopSystem} is \acrfull{GAS} if:
    \begin{itemize}
        \item it is stable in the sense of Lyapunov, i.e. for every $\epsilon > 0$, $\exists \delta(\epsilon) > 0$ such that $\|x(0)\| < \delta \implies \|x(k)\| < \epsilon$ for all $k \in \integersN{}_{\geq 0}$.
        \item it is globally attractive, i.e. $\lim_{k \to \infty} \|x(k)\| = 0$ for all $x(0) \in \realsN{n}$.
    \end{itemize}
\end{defn}
\begin{defn}[Local Asymptotic Stability]
    Closed-loop system \cref{eq:GeneralClosedLoopSystem} is \acrfull{LAS} if
    \begin{itemize}
        \item it is stable in the sense of Lyapunov,
        \item it is locally attractive, i.e. $\exists \eta > 0$ such that $\lim_{k \to \infty} \|x(k)\| = 0$ for all $\|x(0)\| < \eta$.
    \end{itemize}
\end{defn}
In addition, for locally stable systems, we would like to find (an inner estimate of) the closed-loop system's \acrfull{ROA}.
\begin{defn}[Region of Attraction]
    The \acrfull{ROA} of the equilibrium point $x = 0$ of closed-loop system \labelcref{eq:GeneralClosedLoopSystem} is the set $\mathcal{X}_{\textrm{RoA}}$ of all points such that $x \in \mathcal{X}_{\textrm{RoA}}$ implies $\lim_{k\to\infty} \|x(k)\| = 0$.
\end{defn}

A certificate of local/global asymptotic stability and an estimate of the closed-loop system's \acrshort{ROA} can be obtained systematically by formulating and successfully solving one or more \acrshortpl{SDP} searching for a Lyapunov function.

\begin{defn}[Lyapunov Function, Def. B.12 \cite{mybibfile:Rawlings2017}]
    \label{defn:LyapunovFunction}
    Suppose that $\mathcal{X} \subseteq \realsN{n_x}$ is positive invariant and the origin $0\in \mathcal{X}$ is an equilibrium point for closed-loop system \cref{eq:GeneralClosedLoopSystem}. A function $V: \mathcal{X}\rightarrow R_{\geq 0}$ is said to be a Lyapunov function in $\mathcal{X}$ for closed-loop system \cref{eq:GeneralClosedLoopSystem} if there exist functions $\alpha_1, \alpha_2 \in \mathcal{K}_{\infty}$, and continuous, positive definite function $\alpha_3$ such that for any $x \in \mathcal{X}$ 
\end{defn}
\begin{subequations}\label{eq:V_conditions}
    \begin{alignat}{1}
    V(x) &\geq \alpha_1(\|x\|) \label{eq:V_conditions_a}, \\
    V(x) &\leq \alpha_2(\|x\|) \label{eq:V_conditions_b}, \\
     V(x) - V\big(f(x,u(x))\big) &\geq \alpha_3(\|x\|)\label{eq:V_conditions_c}.
\end{alignat}
\end{subequations}
In this work, the Lyapunov function in $\mathcal{X}$ for closed-loop system \cref{eq:GeneralClosedLoopSystem} is referred to as a global Lyapunov function 
if $\mathcal{X} = \realsN{n_x}$, and a local Lyapunov function 
otherwise. In addition, any function $V$ satisfying \cref{eq:V_conditions_a,eq:V_conditions_b} is referred to as a candidate Lyapunov function.

Following from standard Lyapunov theory \cite[Thm~B.13]{mybibfile:Rawlings2017}, given that $f\big(x, \varphi(x)\big)$ is locally bounded, the existence of a global or local Lyapunov function for closed-loop system \cref{eq:GeneralClosedLoopSystem} certifies that the system is \acrshort{GAS} or \acrshort{LAS}, respectively. In addition, the domain of this Lyapunov function is an (inner) estimate of the closed-loop system's \acrshort{ROA}.

The remainder of this section reviews and highlights the limitations of a specific state-of-the-art stability verification procedure in which a semialgebraic-set-based system model of closed-loop system \cref{eq:GeneralClosedLoopSystem} is used to formulate one or more \acrshortpl{SDP} searching for a Lyapunov function. \Cref{sec:SoAReview_SAModeling} reviews how the semialgebraic-set-based system model of closed-loop system $x^+=f\big(x,\varphi(x)\big)$ is obtained.  \Cref{sec:SoAReview_SDPFormulation} details how this system model is used to formulate the \acrshort{SDP}(s) searching for a Lyapunov function. \Cref{sec:SoAReview_LimitationsOfSoA} concludes this section by summarizing the limitations of this state-of-the-art stability verification procedure, which are addressed by our contributions in \cref{sec:ModelingContributions,sec:StabilityVerificationContributions} of this work.

\subsection{Semialgebraic-set-based System Model}
\label{sec:SoAReview_SAModeling}
Following \cite{mybibfile:Korda2022,mybibfile:Newton2021,mybibfile:Newton2022}, closed-loop system \cref{eq:GeneralClosedLoopSystem} is modeled using semialgebraic sets, i.e. sets defined using polynomial (in)equalities. The resulting semialgebraic sets, $\mathbf{K}_\varphi$ and $\mathbf{K}_L$, model \acrshort{NNC} $\varphi$ and composed loop $L = \varphi \circ f \circ (\textrm{id}, \, \varphi)$, shown in \cref{subfig:NeuralNetworkGraph,subfig:OpenLoopGraph}, respectively, by providing semialgebraic descriptions of their graphs.

\begin{defn}[Graph of a function]
    \label{defn:FunctionGraph}
    Given a function $f:\mathcal{X}\rightarrow \mathcal{Y}$, its graph $\graph(f)$ is defined by the set
    \vspace{-0.3\baselineskip}
    \begin{equation}
        \graph({f})=
        \big\{ \big(x,y\big) 
        \; \big\vert \; x \in \mathcal{X}, \, y = f(x)  \big\}.
        \label{eq:GraphDefinition}
    \end{equation}
\end{defn}
\vspace{0.3\baselineskip}

Ideally, the semialgebraic sets $\mathbf{K}_\varphi$, $\mathbf{K}_L$ are identically equal to $\graph(\varphi)$, $\graph(L)$, respectively. This allows the system model to capture the exact input-output relation of \acrshort{NNC} $\varphi$ and composed loop $L$. If such an exact description is not available, the semialgebraic set defines the graph of a multivalued function encompassing the desired function $\varphi$ or $L$, e.g. $\graph(\varphi) \subset \mathbf{K}_\varphi$. System models containing such inexact descriptions fail to uniquely describe how the states of the closed-loop system change over time, which is undesirable as the system model $\big(\mathbf{K}_\varphi, \, \mathbf{K}_{L}\big)$ specifies all state trajectories along which a Lyapunov function must exhibit a decrease.

\subsubsection{Semialgebraic Description of \texorpdfstring{\acrshort{NNC} $\varphi$}{NNC Phi}}
\label{sec:SoAReview_SAModeling_Kphi}
Previous works \cite{mybibfile:Korda2022,mybibfile:Newton2021,mybibfile:Newton2022} have restricted \acrshort{NNC} $\varphi$ to be a fixed, deep, feedforward neural network $\varphi\colon \realsN{n_x} \mapsto \realsN{n_u}$ of the form
\begin{subequations}
\label{eq:FFNNdefinition}
    \begin{alignat}{2}
        \varphi(x) &= A_{\ell+1} \circ \phi_\ell \circ A_{\ell} \circ \, \ldots \circ \phi_2 \circ A_2 \circ \phi_1 \circ A_1(x), \ &&
        \label{eq:FFNNdefinition_Concatenation} \\
        A_i(x) &= W_i x + b_i. \quad && 
        \label{eq:FFNNdefinition_AffineTransformation} 
    \end{alignat}
\end{subequations}
Here, $W_i \in \realsN{n_{i} \times n_{i-1}}$, $b_i \in \realsN{n_i}$, $\phi_i \colon \realsN{n_i} \mapsto \realsN{n_{i}}$ represent the weights, biases and stacked activation functions of layer $i$, respectively, and $n_i$ denotes the number of inputs to (hidden) layer $i+1$. 

To obtain a semialgebraic-set-based model of feedforward \acrshort{NNC} $\varphi$, following ideas from robust control, the outputs of the nonlinear components of \cref{eq:FFNNdefinition} are first isolated and collected in a so-called `lifting variable', $\lambda = [\lambda_1\transpose \, , \ldots , \, \lambda_\ell\transpose]\transpose \in \realsN{n_\lambda}$, with

\begin{equation}
    \label{eq:FFNNLifted}
    \begin{alignedat}{1}
        \lambda_1(x) &= \phi_1(W_1x + b_1), \\
        \lambda_2(x) &= \phi_2(W_2\lambda_1(x) + b_2), \\
        \vdots \quad &= \qquad \vdots \\
        \lambda_\ell(x) &= \phi_\ell(W_\ell\lambda_{\ell-1}(x) + b_\ell).
    \end{alignedat}
\end{equation}
In this work, when utilized, such lifting variables are interpreted as auxiliary outputs of the main function to be described, i.e. here semialgebraic set $\mathbf{K}_\varphi$ aims to characterize $\graph\big(\bigl[\!\begin{smallmatrix} \lambda(x) \\ \varphi(x) \end{smallmatrix}\!\bigr]\big)$, as shown in \cref{subfig:NeuralNetworkGraph}.

In contrast to robust control, where the input-output relation of all $\phi_i$ is typically assumed to only be known approximately, the individual nonlinear components $\phi_i$ of a \acrshort{NNC} are (stacked) activation functions whose definitions are known exactly. This motivates the use of semialgebraic sets to model \acrshort{NNC} $\varphi$ instead of conventional techniques from robust control such as integral/incremental constraints or Zames-Falb multipliers. 

In order to obtain a semialgebraic description of feedforward \acrshort{NNC} $\varphi$, the graphs of the individual activation functions $\phi_i$ are first examined. The graphs of some common activation functions, such as \acrshort{ReLU} and the saturation function, can be described exactly using a semialgebraic set \cite{mybibfile:Korda2022}.

\begin{exmp}
    \label{exmp:SAset_SingleReLU}
    Consider the composition of the \acrshort{ReLU} activation function with an affine transformation $A$, $\realsN{n} \ni x \mapsto y = \phi_{\textrm{ReLU}} \circ A(x) = \max(0, \, w\transpose x + b)$. The graph of this function is a semialgebraic set, as 
    $\graph(\phi_{\textrm{ReLU}} \circ A) =$
\begin{equation}
    \label{eq:SingleReLUNetwork}
    \mathbf{K}_{\phi_{\textrm{ReLU}} \circ A} = \left\{ (x, y) 
    \, \left\vert \:
    {\renewcommand{\arraystretch}{0.9}%
    \begin{alignedat}{1}
        & y \geq 0, \ y - w\transpose x - b \geq 0 \\
        \span y \big(y - w\transpose x - b \big) = 0
    \end{alignedat}
    }
    \right. \right\}.
\end{equation}
\end{exmp}
\vspace{0.3\baselineskip}

Other common activation functions are transcendental, e.g. $\phi_{\textrm{sp}}(x) = \ln(1+e^{x})$ or $\phi_{\tanh}(x) = \tfrac{e^x - e^{-x}}{e^x + e^{-x}}$. Such functions do not satisfy a polynomial equation, and therefore polynomial (in)equalities are used to define inexact descriptions of their graphs
\cite{mybibfile:Newton2021}, e.g. via (multiple) sector constraints, slope constraints, or robust control descriptions such as integral/incremental constraints and Zames-Falb multipliers. The resulting inexact semialgebraic set description 
defines the graph of a multivalued function enclosing the desired activation function.

\begin{exmp}
    \label{exmp:SAset_SingleSigmoid}
     Consider the $\tanh$ activation function. 
     Using a single sector constraint and the asymptote values, the graph of $\realsN{n} \ni x \mapsto y =\phi_{\textrm{tanh}}(x) = \tfrac{e^x - e^{-x}}{e^x + e^{-x}}$ can be approximated by
    \begin{equation}
        \label{eq:SigmoidGraphApproximation}
        \mathbf{K}_{\phi_{\textrm{tanh}}} = \left\{(x, y) \, \Bigg\vert 
        {\renewcommand{\arraystretch}{0.9}%
            \begin{bmatrix}
            y + 1 \\
            1 - y \\
            y(x-y)
            \end{bmatrix} \geq 0
        }
        \right\}.
    \end{equation}
\end{exmp}
\vspace{0.4\baselineskip}

Once a semialgebraic-set-based description of each activation function's graph is available, the semialgebraic description of the complete feedforward \acrshort{NNC} $\varphi$ is obtained by
examining the composition of single-neuron graphs across successive layers of the network, resulting in a semialgebraic-set-based description of \acrshort{NNC} $\varphi$ defined by

\begin{equation}
    \label{eq:SemialgebraicNetworkSet}
         \mathbf{K}_\varphi = \Bigg\{ \bigg(x, \begin{bmatrix} \lambda \\ \varphi \end{bmatrix} \bigg) 
         \;  \bigg\vert \; 
                g^{\varphi}(x, \lambda, \varphi) \geq 0, \,
            h^{\varphi}(x, \lambda, \varphi) = 0  \Bigg\},
\end{equation}
where $g^{\varphi} \colon \realsN{n_x} \times \realsN{n_\lambda} \times \realsN{n_u} \mapsto \realsN{n_{g^\varphi}}$ and $h^{\varphi} \colon \realsN{n_x} \times \realsN{n_\lambda} \times \realsN{n_u} \mapsto \realsN{n_{h^\varphi}}$ are (vector-valued) polynomial functions.

\begin{exmp}
    \label{exmp:SAset_CompleteReLU}
    Assume \cref{eq:FFNNdefinition} utilizes only \acrshort{ReLU} activation functions. Then, application of \cref{eq:SingleReLUNetwork} to all neurons in \acrshort{NNC} $\varphi$ leads to
    \begin{equation}
    \label{eq:RecursiveReLUSemialgebraicSet}
    \left\{ (\lambda_{i-1}, \lambda_i) 
    \, \left\vert \:
    \begin{alignedat}{1}
        & \lambda_i \geq 0, \ \lambda_i - W_i \lambda_{i-1} - b_i \geq 0 
        \\ & \lambda_i \odot \left(\lambda_i - W_i \lambda_{i-1} - b_i \right) = 0
    \end{alignedat}
    \right. \right\}, \quad \forall i \in [\ell],
    \end{equation}
    where $\lambda_0 = x$ only here for notational purposes. An exact semialgebraic description $\mathbf{K}_\varphi$ of the graph $\realsN{n} \ni x \mapsto [\lambda(x)\transpose, \, \varphi(x)\transpose]\transpose$ is obtained by considering the union of these sets with the $n_u$ polynomial equalities $\varphi - W_{\ell+1} \lambda_\ell - b_{\ell+1} = 0$.    
    \end{exmp}
\vspace{0.3\baselineskip}

As is clear from \cref{exmp:SAset_SingleSigmoid}, any set $\mathbf{K}_\varphi$ constructed using inexact semialgebraic-set-based descriptions of individual activation functions $\phi_i$ defines a system model specifying multiple admissible outputs $\bigl[ \begin{smallmatrix} \lambda \\ \varphi \end{smallmatrix} \bigr]$ for any input $x \in \realsN{n_x}$.

\subsubsection{Semialgebraic Description of Composed Loop \texorpdfstring{L}{L}}
\label{sec:SoAReview_SAModeling_KL}
A semialgebraic description of the graph of composed loop $L = \varphi \circ f \circ (\textrm{id}, \varphi)$ is obtained by examining the function $\realsN{n} \ni x \mapsto [\lambda(x)\transpose, \, \allowdisplaybreaks \varphi(x)\transpose, \, \allowdisplaybreaks {x^+}(x)\transpose, \, \allowdisplaybreaks \lambda^+(x)\transpose, \, \allowdisplaybreaks \varphi^+(x)\transpose]\transpose$, shown in \cref{subfig:OpenLoopGraph}. The semialgebraic set describing the graph of this function is constructed by considering an additional composition with the polynomial dynamics $f$, which is achieved via the inclusion of the equalities $x^+ - f(x, \varphi) = 0$. Thus, following the technique used to construct $\mathbf{K}_\varphi$, consider the semialgebraic description of $L$ and its auxiliary outputs,
\begin{equation}
    \label{eq:SemialgebraicComposedLoopSet}  
    \begin{aligned}
         & \mathbf{K}_L =  \Bigg\{ 
        \Big(x, \, [\lambda\transpose\!\!, \, \varphi\transpose\!\!, \, ({x}^{+})\transpose\!\!\!\!, \,
        (\lambda^+)\transpose\!\!\!\!, \, ({\varphi^+})\transpose]\transpose\Big) 
        \;  \bigg\vert
        \\  & \qquad \qquad \qquad \qquad \quad
        \begin{aligned}
             g^L(x,  \, \lambda, \, \varphi, \, x^+, \, \lambda^+, \, \varphi^+) & \geq 0, \\
             h^L(x, \, \lambda, \, \varphi, \, x^+, \, \lambda^+, \, \varphi^+) &= 0
        \end{aligned} \, 
        \Bigg\},
    \end{aligned}
\end{equation}
with $g^L \colon \realsN{n_x} \! \times \! \realsN{n_\lambda} \! \times \! \realsN{n_u} \! \times \! \realsN{n_x} \! \times \! \realsN{n_\lambda} \! \times \! \realsN{n_u} \mapsto \realsN{n_{g^L}}$ and $h^L \colon \realsN{n_x} \! \times \! \realsN{n_\lambda} \! \times \! \realsN{n_u} \! \times \! \realsN{n_x} \! \times \! \realsN{n_\lambda} \! \times \! \realsN{n_u} \mapsto \realsN{n_{h^L}}$ (vector-valued) polynomial functions. 

By construction, the semialgebraic-set-based system model 
$(\mathbf{K}_\varphi, \mathbf{K}_L)$
describes the input-output relations of \acrshort{NNC} $\varphi$ and composed loop $L$ for all $x \in \realsN{n}$, respectively, and does so exactly in case all neurons in feedforward \acrshort{NNC} $\varphi$ are modeled using an exact semialgebraic description. 

To simplify the notation in the remainder of this work, we now define the vectors
$\zeta = [\zeta_x\transpose, \, \zeta_\lambda\transpose, \, \zeta_\varphi\transpose]\transpose \in \realsN{n_x+n_\lambda+n_u} = \realsN{n_\zeta}$ and $\xi = [\xi_x\transpose, \, \xi_\lambda\transpose, \, \xi_\varphi\transpose, \, \xi_{x^+}\transpose, \, \xi_{\lambda^+}\transpose, \, \xi_{\varphi^+}\transpose]\in \realsN{2(n_x + n_\lambda + n_u)} = \realsN{n_\xi}$.

\subsection{SDP formulation}
\label{sec:SoAReview_SDPFormulation}
Next, using the semialgebraic-set-based system model $\big( \mathbf{K}_\varphi, \, \mathbf{K}_L \big)$, the \acrshort{SDP}(s) searching for a
Lyapunov function for closed-loop system \cref{eq:GeneralClosedLoopSystem} are formulated.
The parameterization of the solution space and the formulation of all \acrshort{SDP} problems are discussed below.

\subsubsection{Parameterization of Solution Space}
\label{sec:SOAReview_SDPFormulation_VParameterization}
In order to pose the search for a Lyapunov function 
as an \acrshort{SDP}, the search domain for this Lyapunov function is defined as a finite-dimensional function space using \acrshort{SOS} polynomials.

\begin{defn}[\Acrfull{SOS} polynomial]
    \label{defn:SOSpolynomial}
    A polynomial $\sigma(\xi)$ is \acrshort{SOS} if it admits a decomposition as a sum of squared polynomials. Let $\nu(\xi)$ denote a vector of monomial terms generated by entries of $\xi$ and let $L$ represent a matrix of coefficients. \acrshort{SOS} polynomials admit an equivalent semidefinite representation
    \begin{equation}
        \label{eq:SOSDefinition}
        \sigma(\xi) = \sum_i \sigma_i(\xi)^2 = \sum_i (l_i\transpose \nu(\xi))^2 = \nu(\xi)\transpose \underbrace{L\transpose L}_{\succeq 0} \nu(\xi),       
    \end{equation}    
    allowing \acrshort{SOS} polynomials to be used in \acrshortpl{SDP}. See the work of Parrilo \cite{mybibfile:Parrilo2003} for more information.
\end{defn}

Following the modeling procedure of \cref{sec:SoAReview_SAModeling}, the semialgebraic set $\mathbf{K}_{\varphi}$ specifies the model-admissible values of $\bigl[ \begin{smallmatrix} \lambda \\ \varphi \end{smallmatrix} \bigr]$ for any input $x \in \realsN{n_x}$. Therefore, in addition to the state vector $x$, lifting variables $\lambda$ and control output $\varphi$ are used to parameterize the finite-dimensional function space that is explored in the \acrshort{SDP}-based search for a Lyapunov function. Following \cref{defn:LyapunovFunction}, to ensure that the function space consists of non-negative functions for all $\big(x, \bigl[ \begin{smallmatrix} \lambda \\ \varphi \end{smallmatrix} \bigr]\big) \in \graph(\varphi)\subseteq \mathbf{K}_\varphi$, the space of functions $V \colon \realsN{n_x} \mapsto \realsN{}_{\geq 0}$ is parameterized as

\begin{equation}
{\scalebox{0.925}{$
    \begin{aligned}
        \label{eq:LyapunovParam}
        V(x) &= V_{\zeta}(x, \lambda(x), \varphi(x))
        \\ 
        & 
        \begin{multlined}
            = \sigma^V\big(x, \lambda(x), \varphi(x)\big) 
            \\
            + \sigma^V_{\textrm{ineq}}(x, \lambda(x), \varphi(x))\transpose 
            \underbrace{\begin{bmatrix}
                \mathcal{M}\big(g^\varphi(x, \lambda(x), \varphi(x)), 1\big) \\ 
                \mathcal{M}(g^\varphi(x, \lambda(x), \varphi(x)), 2\big) \\ 
                \vdots
            \end{bmatrix}}_{g^V},
        \end{multlined}
    \end{aligned}
$}}
\end{equation}
with parameters $\sigma^V$, $\sigma^V_{\textrm{ineq}}$ representing any scalar \acrshort{SOS} polynomial and any vector of \acrshort{SOS} polynomials, respectively. This class of functions is assumed continuous in $x$.

\begin{assume}
    \label{assume:ContinuityOfV}
    All functions $V$ parameterized by \cref{eq:LyapunovParam} are continuous functions of $x$ for all $x \in \realsN{n_x}$.
\end{assume}

For feedforward \acrshortpl{NNC} $\varphi$ of \cref{eq:FFNNdefinition} this assumption is generally non-restrictive as $\lambda(x)$ and $\varphi(x)$ are continuous function of $x$ for 
continuous stacked activation functions $\phi_1, \ldots, \phi_\ell$.

The solution space parameterized by \cref{eq:LyapunovParam} is now combined with the semialgebraic system model $\big(\mathbf{K}_\varphi, \, \mathbf{K}_L\big)$ to formulate \acrshortpl{SDP} posing the search for a 
Lyapunov function for closed-loop system \cref{eq:GeneralClosedLoopSystem} as an optimization problem.

\subsubsection{Global Lyapunov Function}
\label{sec:SOAReview_SDPFormulation_GlobalStability}
Following \cref{eq:V_conditions_c}, the value of any valid global Lyapunov function $V$ for closed-loop system \cref{eq:GeneralClosedLoopSystem} must strictly decrease over time for all $x \in \realsN{n} \setminus \{0\}$. 
To guarantee this required decrease for a function $V$ parameterized by \cref{eq:LyapunovParam} and all model-admissible input-output pairs $\big(x, \, [\lambda\transpose\!\!, \, \varphi\transpose\!\!, \,({x}^{+})\transpose\!\!\!\!, \, (\lambda^+)\transpose\!\!\!\!, \, ({\varphi^+})\transpose]\transpose\big) \in \mathbf{K}_L$, 
it is sufficient to show that 
\begin{equation}
{\scalebox{0.925}{$
\begin{aligned}
    & V_{\zeta}(\xi_x, \xi_\lambda, \xi_\varphi) - V_{\zeta}(\xi_{x^+}, \xi_{\lambda^+}, \xi_{\varphi^+}) - \|\xi_x\|^2 - p^{\Delta V}_{\textrm{eq}}(\xi)\transpose h^L(\xi)
    \\ 
    & \ \  - \sigma^{\Delta V}(\xi) - \sigma^{\Delta V}_{\textrm{ineq}}(\xi) \transpose
    \begin{bmatrix}
        \mathcal{M}\big(g^L(\xi), 1\big) \\
        \mathcal{M}\big(g^L(\xi), 2\big) \\
        \vdots
    \end{bmatrix}
      \geq 0, \ \forall \xi 
      \in \realsN{n_{\xi}}, 
    \label{eq:LyapunovDecreaseCond}
\end{aligned}
$}}
\end{equation}
with $p^{\Delta V}_{\textrm{eq}}$ any vector of arbitrary polynomials, $\sigma^{\Delta V}$ any scalar \acrshort{SOS} polynomial and $\sigma^{\Delta V}_{\textrm{ineq}}$ any vector of \acrshort{SOS} polynomials. By substituting the variables $\xi_x$, $\xi_\lambda$, $\xi_\varphi$, $\xi_{x^+}$, $\xi_{\lambda^+}$, $\xi_{\varphi^+}$ in \cref{eq:LyapunovDecreaseCond} with $x$, $\lambda(x)$, $\varphi(x)$, ${x^+}(x)$, ${\lambda^+}(x)$, ${\varphi^+}(x)$, respectively, 
it follows that satisfaction of \cref{eq:LyapunovDecreaseCond} is a sufficient condition to guarantee the required decrease of $V$ with $\alpha_3(\|x\|) \geq \|x\|^2$ along all state trajectories defined admissible by semialgebraic-set-based system model $\mathbf{K}_L$.

A sufficient, numerically-efficient condition to ensure the satisfaction of \cref{eq:LyapunovDecreaseCond} is to check if it itself is \acrshort{SOS}, e.g. bringing \cref{eq:LyapunovDecreaseCond} to a canonical form 
\begin{equation}
 {\nu_{\textrm{tot}}^{\Delta V}}(\xi)\transpose Q_{\textrm{tot}}^{\Delta V} \nu_{\textrm{tot}}^{\Delta V}(\xi) \geq 0, \ \forall \xi \in \realsN{n_\xi},   
\end{equation}
and verifying $Q_{\textrm{tot}}^{\Delta V} \succeq 0$. Interpreting \cref{eq:LyapunovDecreaseCond} as such an \acrshort{SOS} constraint leads to the \acrshort{SDP} formulation

\begin{subequations}
    \label{eq:SDPFormulationGlobalAsymptoticStability}
    \begin{alignat}{4}
        &\span\span \text{find:} \ & \sigma^V, \, \sigma^V_{\textrm{ineq}}, \ \ \, \,
        & \!\!\! \! \! \! \! \sigma^{\Delta V},  \, \sigma^{\Delta V}_{\textrm{ineq}}, \,  p^{\Delta V}_{\textrm{eq}} \span \nonumber \\
        &\span\span \text{s.t.} \ & \eqref{eq:LyapunovParam}, \ & \eqref{eq:LyapunovDecreaseCond}, \span \\
        &\span\span & \sigma^V, \sigma^V_{\textrm{ineq}}, \sigma^{\Delta V}, \sigma^{\Delta V}_{\textrm{ineq}} \quad & \text{\acrshort{SOS} polynomials,} \\
        &\span\span & 
        p^{\Delta V}_{\textrm{eq}} \quad & \text{arbitrary polynomials,}
    \end{alignat}
\end{subequations}
whose solution defines a global Lyapunov function for closed-loop system \cref{eq:GeneralClosedLoopSystem}.

\begin{thm}
    \label{thm:GASSDPProof}
    Under \cref{assume:ContinuityOfV}, any solution to \acrshort{SDP} \cref{eq:SDPFormulationGlobalAsymptoticStability} defines a global Lyapunov function for closed-loop system \cref{eq:GeneralClosedLoopSystem}.
\end{thm}
\begin{proof}
    Korda \cite[Thm. 2]{mybibfile:Korda2022} presents a direct proof of \acrshort{GAS} given a solution to \acrshort{SDP} \cref{eq:SDPFormulationGlobalAsymptoticStability}, and refers to $V$ as a Lyapunov function in the remainder of his work.

    For an additional formal proof that any function $V$ defined by a solution to \acrshort{SDP} \cref{eq:SDPFormulationGlobalAsymptoticStability} defines a global Lyapunov function according to \cref{defn:LyapunovFunction}, consider the proof of \cref{thm:ValidityOptimProblemInvariantLocalAsymptoticStability} applied to solutions of \acrshort{SDP} \cref{eq:SDPFormulationGlobalAsymptoticStability} where i) $\mathcal{Q}$ should be interpreted as $\realsN{n_x}$, and ii) implications following from constraint \cref{eq:LocalLyapunovDecreaseCond} now follow from \cref{eq:LyapunovDecreaseCond}. 
\end{proof}

\subsubsection{Local Lyapunov Function}
\label{sec:SOAReview_SDPFormulation_LocalStability}
By \cref{defn:LyapunovFunction}, the value of a local Lyapunov function is only required to decrease over time for a subset of all possible state trajectories. To define this subset, 
the set $\mathcal{Q}\subseteq\realsN{n_x}$ is introduced, defined as
\begin{equation}
    \mathcal{Q} = \big\{ x \in \realsN{n} \mid q(x) = q_\zeta(x, \lambda(x), \varphi(x)) \geq 0 \big\}
    \label{eq:LocalRegionQDefinition}
\end{equation}
with $0 \in \textrm{int}(\mathcal{Q})$ and $q_\zeta\colon \realsN{n_x} \times \realsN{n_\lambda} \times \realsN{n_u} \mapsto \realsN{n_q}$ a polynomial function. The function $q$ is assumed continuous.

\begin{assume}
\label{assume:ContinuityOfQinx}
    The function $q(x)$ characterizing $\mathcal{Q}$ is a continuous function of $x$ for all $x\in\realsN{n_x}$.
\end{assume}

Following \cref{defn:FunctionGraph}, the description $\mathbf{K}_{L}$ of composed loop $L$ can be limited to $x \in \mathcal{Q}$ by including the polynomial inequality $q_\zeta(x, \lambda, \varphi) \geq 0$ in $\mathbf{K}_L$. Applying the methodology used to construct \acrshort{SDP} \cref{eq:SDPFormulationGlobalAsymptoticStability}, the inclusion of the constraint $x \in \mathcal{Q}$ results in a modification of \cref{eq:LyapunovDecreaseCond} to

\begin{equation}
{\scalebox{0.925}{$
\begin{aligned}
    &V_\zeta(\xi_{x}, \xi_\lambda, \xi_\varphi) - V_\zeta(\xi_{x^+}, \xi_{\lambda^+}, \xi_{\varphi^+}) - \|\xi_x\|^2 -  p^{\Delta V}_{\textrm{eq}}(\xi)\transpose h^L(\xi)
    \\
    & \ \  - \sigma^{\Delta V}(\xi) - \sigma^{\Delta V}_{\textrm{ineq}}(\xi) \transpose 
     \begin{bmatrix}
        \mathcal{M}\Big( \left[\!\begin{smallmatrix} q_\zeta(\xi_{x}, \xi_\lambda, \xi_\varphi) \\ g^L(\xi) \end{smallmatrix}\!\right]^{\phantom{|}}\!\!, 1 \Big)
        \\
        \mathcal{M}\Big( \left[\!\begin{smallmatrix} q_\zeta(\xi_{x}, \xi_\lambda, \xi_\varphi) \\ g^L(\xi) \end{smallmatrix}\!\right]^{\phantom{|}}\!\!, 2 \Big)
        \\
        \vdots
    \end{bmatrix} 
    \geq 0, \  \forall \xi 
    \in \realsN{n_{\xi}},
    \label{eq:LocalLyapunovDecreaseCond}
\end{aligned}
$}}
\end{equation}
which is similarly interpreted as an \acrshort{SOS} constraint in the 
\acrshort{SDP}
\begin{subequations}
    \label{eq:SDPFormulationLocalAsymptoticStability}
    \begin{alignat}{4}
        &\span\span \text{find:} \ & \sigma^V, \, \sigma^V_{\textrm{ineq}}, \ \ \, \,
        & \!\!\! \! \! \! \! \sigma^{\Delta V},  \, \sigma^{\Delta V}_{\textrm{ineq}}, \,  p^{\Delta V}_{\textrm{eq}} \span \nonumber \\
        &\span\span \text{s.t.} \ & \eqref{eq:LyapunovParam}, \ & \eqref{eq:LocalLyapunovDecreaseCond}, \span \\
        &\span\span & \sigma^V, \sigma^V_{\textrm{ineq}}, \sigma^{\Delta V}, \sigma^{\Delta V}_{\textrm{ineq}} \quad & \text{\acrshort{SOS} polynomials,} \label{eq:LocalAsymptoticStabilitySOS} \\
        &\span\span & 
        p^{\Delta V}_{\textrm{eq}} \quad & \text{arbitrary polynomials.} \label{eq:LocalAsymptoticStabilityP}
    \end{alignat}
\end{subequations}

However, the set $\mathcal{Q}$ is not known a priori to be invariant, and therefore a general solution to \acrshort{SDP} \labelcref{eq:SDPFormulationLocalAsymptoticStability} is not guaranteed to define a valid Lyapunov function (see \cref{exmp:CounterexampleLAS} of \cref{sec:StabilityVerificationContributions}). To recover this guarantee, the general solution space defined by \cref{eq:LyapunovParam} is heavily restricted to a small class of known candidate Lyapunov functions, e.g. $V(x) = x\transpose P x$, $P \succ 0$ \cite{mybibfile:Korda2017,mybibfile:Yin2022,mybibfile:Pauli2021,mybibfile:Newton2022}, which reduces the space of potential solutions for the \acrshort{SDP}.

Once a solution to \acrshort{SDP} \cref{eq:SDPFormulationLocalAsymptoticStability} has been obtained, certifying that this solution defines a valid Lyapunov function for closed-loop system \cref{eq:GeneralClosedLoopSystem} involves finding an invariant set $\mathcal{X} \subseteq \mathcal{Q}$, which will then form an estimate of the system's \acrshort{ROA}.
By \cref{eq:LocalLyapunovDecreaseCond}, a potential trivial class of invariant sets in $\mathcal{Q}$ consists of sublevel sets of $V$, $\mathcal{L}_\gamma(V) \coloneq \big\{ x \in \realsN{n} \mid V(x) \leq \gamma \big\}$. Therefore, a second optimization problem is formulated to find the largest sublevel set of $V$ contained in $\mathcal{Q}$.
Consider the condition
\begin{multline}
    \sigma^{\mathcal{Q}}_q(\zeta) q_\zeta(\zeta) \geq p^{\mathcal{Q}}_{\textrm{eq}}(\zeta)\transpose h^\varphi(\zeta) + \sigma^{\mathcal{Q}}(\zeta) + \sigma^{\mathcal{Q}}_V(\zeta) \big( \gamma - V_\zeta(\zeta) \big) \\ + \sigma^{\mathcal{Q}}_{\textrm{ineq}}(\zeta)\transpose  \begin{bmatrix}
        \mathcal{M}\big(g^\varphi(\zeta), 1 \big) \\
        \mathcal{M}\big(g^\varphi(\zeta), 2 \big) \\
        \vdots        
    \end{bmatrix}, \  \forall \zeta 
    \in \realsN{n_{\zeta}},
    \label{eq:SublevelSetCond}
\end{multline}
with $p^{\mathcal{Q}}_{\textrm{eq}}$ a vector of arbitrary polynomials, $\sigma^{\mathcal{Q}}_q$, $\sigma^{\mathcal{Q}}_V$ and $\sigma^{\mathcal{Q}}$ scalar \acrshort{SOS} polynomials, $\sigma^{\mathcal{Q}}_{\textrm{ineq}}$ a vector of \acrshort{SOS} polynomials, and $V_\zeta$ given by a solution to \acrshort{SDP} \cref{eq:SDPFormulationLocalAsymptoticStability}. By substituting the variables $\zeta_x$, $\zeta_\lambda$, $\zeta_\varphi$ in \cref{eq:SublevelSetCond} with $x$, $\lambda(x)$, $\varphi(x)$, respectively, it follows that satisfaction of \cref{eq:SublevelSetCond} is a sufficient condition to guarantee $q(x) \geq 0$ for all $x\in\mathcal{L}_{\gamma}(V)$, i.e. $\mathcal{L}_\gamma(V) \subseteq \mathcal{Q}$.
Viewing \cref{eq:SublevelSetCond} as an \acrshort{SOS} constraint leads to the optimization problem
\begin{subequations}
    \label{eq:SDPFormulationLargestSublevelSet}
    \begin{alignat}{4}
        &\span\span  \underset{ \substack{ \gamma, \; \sigma^{\mathcal{Q}}_q, \, \sigma^{\mathcal{Q}}, \\ \sigma^{\mathcal{Q}}_V, \, \sigma^{\mathcal{Q}}_\textrm{ineq}, \, p^{\mathcal{Q}}_{\textrm{eq}}}}{\textrm{maximize:}} \ &   & \! \! \! \! \! \! \! \! \gamma  \span \nonumber \\
        &\span\span \text{s.t.} \ &  & \! \! \! \! \! \! \! \! \! \! \! \! \!   \eqref{eq:SublevelSetCond}, \span \label{eq:SDPFormulationLargestSublevelSet_SublevelSetCond} \\
        &\span\span & \sigma^{\mathcal{Q}}_q, \sigma^{\mathcal{Q}}, \sigma^{\mathcal{Q}}_V, \sigma^{\mathcal{Q}}_{\textrm{ineq}} \quad & \text{\acrshort{SOS} polynomials,} \label{eq:test_2} \\
        &\span\span & p^{\mathcal{Q}}_{\textrm{eq}} \quad & \text{arbitrary polynomials,} \label{eq:test_3}
    \end{alignat}
\end{subequations}
which can be solved as an \acrshort{SDP} by i) taking $\sigma^{\mathcal{Q}}_V$ to be a fixed \acrshort{SOS} polynomial, e.g. $\|\zeta\|^2$, or ii) applying a line search method on $\gamma$. Thus, consecutively solving \acrshortpl{SDP} \cref{eq:SDPFormulationLocalAsymptoticStability,eq:SDPFormulationLargestSublevelSet} with $V(0) < \gamma$ defines a local Lyapunov function $V$ in $\mathcal{L}_\gamma(V)$ for closed-loop system \cref{eq:GeneralClosedLoopSystem}.

\begin{thm}
    \label{thm:LASSDPProof}
     Under \cref{assume:ContinuityOfV}, any set of consecutive solutions to \acrshortpl{SDP} \cref{eq:SDPFormulationLocalAsymptoticStability,eq:SDPFormulationLargestSublevelSet} satisfying $V(0) < \gamma$ defines a local Lyapunov function in $\mathcal{L}_\gamma(V)$ for closed-loop system \cref{eq:GeneralClosedLoopSystem}, proving $\mathcal{L}_\gamma(V)$ forms part of this system's \acrshort{ROA}.
\end{thm}
\begin{proof}
    By construction, consecutive solutions to \acrshortpl{SDP} \cref{eq:SDPFormulationLocalAsymptoticStability,eq:SDPFormulationLargestSublevelSet} satisfying $V(0) < \gamma$ verify $\mathcal{L}_\gamma(V)$ is an invariant set with $0 \in \textrm{int}\big(\mathcal{L}_\gamma(V)\big)$. Application of the proof presented by Korda \cite[Thm. 2]{mybibfile:Korda2022} shows how closed-loop system \cref{eq:GeneralClosedLoopSystem} is \acrshort{LAS}. In addition, 
    Korda refers to $V$ as a Lyapunov function in his work.

    For an additional formal proof that any solution $V$ defined by a solution to \acrshort{SDP} \cref{eq:SDPFormulationGlobalAsymptoticStability} constitutes a local Lyapunov function in $\mathcal{L}_\gamma(V)$ according to \cref{defn:LyapunovFunction}, consider the proof of \cref{thm:ValidityOptimProblemInvariantLocalAsymptoticStability} applied to solutions of \acrshort{SDP} \cref{eq:SDPFormulationLocalAsymptoticStability} where $\mathcal{Q}$ should be interpreted as $\mathcal{L}_\gamma(V)$ defined by the solution to \acrshort{SDP} \cref{eq:SDPFormulationLargestSublevelSet}, which is by construction invariant. 
\end{proof}

\subsection{Limitations of the State-of-the-Art Procedure}
\label{sec:SoAReview_LimitationsOfSoA}
Whilst the current state-of-the-art procedure reviewed in this section has shown promising results in both global \cite{mybibfile:Korda2022} and local \cite{mybibfile:Newton2022} stability analyses for \acrshort{NNC} consisting of up to several hundred neurons \cite{mybibfile:Korda2022}, this stability verification method also suffers from several limitations, including a)  the inexact semialgebraic system model of any \acrshort{NNC} utilizing transcendental activation functions, b) the assumed feedforward architecture of \acrshort{NNC} $\varphi$, c) the heavily restricted solution space when solving \acrshort{SDP} \cref{eq:SDPFormulationLocalAsymptoticStability} in the search for a local Lyapunov function, and d) the indirect optimization of the \acrshort{ROA} estimate when solving \acrshortpl{SDP} \cref{eq:SDPFormulationLocalAsymptoticStability,eq:SDPFormulationLargestSublevelSet}. 

The remainder of this paper presents our contributions designed to address these limitations.
\Cref{sec:ModelingContributions} details improvements to the semialgebraic-set-based modeling procedure to improve the fidelity of semialgebraic-set-based models of \acrshortpl{NNC} and expands the class of \acrshort{NNC} that are suitable for analysis. \cref{sec:StabilityVerificationContributions} introduces improved \acrshort{SDP} formulations that ease the restriction on the solution space when searching for a local Lyapunov function and allow the size of the \acrshort{ROA} estimate to be optimized directly.

\section{Improved Modeling with Semialgebraic Sets}
\label{sec:ModelingContributions}
In this section we present our contributions to the semialgebraic-set-based modeling procedure reviewed in \cref{sec:SoAReview_SAModeling}, which expand the class of \acrshortpl{NNC} that can be described exactly using a semialgebraic-set-based model.
\Cref{sec:NewActivationFunctions} details novel, alternative activation functions that share the fundamental properties of transcendental activation functions $\text{softplus}$ and $\text{tanh}$ whilst admitting an exact semialgebraic representation. \Cref{sec:NewNNArchitectures} relaxes the assumption that \acrshort{NNC} $\varphi$ is an $\ell$-layer feedforward network described by \cref{eq:FFNNdefinition}, and shows that the semialgebraic-set-based modeling procedure can still be applied if the \acrshort{NNC} is part of the more general class of \acrfullpl{REN}.

\begin{figure*}[t]
    \centering
    \begin{subfigure}{0.47\textwidth} 
        \centering
         \resizebox{\linewidth}{!}{ 
\begin{tikzpicture}
    \draw[step=1cm, gray!30, very thin] (-5.5, -0.5) grid (5.5, 5.5);
    
    \draw[thick, ->] (-5.5, 0) -- (5.5, 0) node[above=0.125cm, xshift=-0.15cm] {$x$};
    \draw[thick, ->] (0, -0.5) -- (0, 5.5); 
    
    \draw[domain=-5:5, smooth, thick, samples=100, color=blue] 
        plot (\x, {ln(1 + exp(\x))}); 

    \draw[domain=-5:5, smooth, thick, samples=100, color=red] 
        plot (\x, {\x/2 + sqrt(0.480453 + (\x)^2/4)}); 

    \foreach \x in {-5, -4, -3, -2, -1, 1, 2, 3, 4, 5} 
        \draw (\x, 0.1) -- (\x, -0.1) node[below, fill=white] {\x};
    \foreach \y in {1, 2, 3, 4, 5} 
        \draw (0.1, \y) -- (-0.1, \y) node[left, fill=white] {\y};

    \begin{scope}[shift={(-5, 4)}] 
        \node[draw, fill=white, inner sep=5pt, anchor=west, scale=0.8] at (0, 0) { \begin{minipage}{2.5cm} 
            \begin{tikzpicture}[baseline]
                \draw[blue, thick] (0, 0) -- (0.6, 0); 
                \node[anchor=west] at (0.8, 0) {$\ln(1 + e^x)$}; 
            \end{tikzpicture}
    
            \begin{tikzpicture}[baseline]
                \draw[red, thick] (0, 0) -- (0.6, 0); 
                \node[anchor=west] at (0.8, 0) {$\hat{\phi}_{\textrm{sp}}(x)$}; 
            \end{tikzpicture}
        \end{minipage}
        };

    \end{scope}
\end{tikzpicture}
}
        \caption{}
        \label{subfig:softplus_approximation}
    \end{subfigure}
    \hspace{0.03\textwidth}
    \begin{subfigure}{0.47\textwidth} 
      \centering
      \resizebox{\linewidth}{!}{ 
\begin{tikzpicture}
    \draw[step=1cm, gray!30, very thin] (-5.5, -2.8) grid (5.5, 2.8);
    
    \draw[thick, ->] (-5.5, 0) -- (5.5, 0) node[above=0.125cm, xshift=-0.15cm] {$x$};
    \draw[thick, ->] (0, -3) -- (0, 3); 
    
    \draw[domain=-5:5, smooth, thick, samples=100, color=blue] 
        plot (\x, {2*(exp(\x) - exp(-\x))/((exp(\x) + exp(-\x))}); 

    \draw[domain=-5:5, smooth, thick, samples=100, color=red] 
        plot (\x, {2*(1.171*\x) / (sqrt(1 + (1.171*\x)^2))}); 



    \foreach \x in {-5, -4, -3, -2, -1, 1, 2, 3, 4, 5} 
        \draw (\x, 0.1) -- (\x, -0.1) node[below, fill=white] {\x};
    \foreach \y in {-1, -0.5} 
        \draw (-0.1, {2*\y}) -- (0.1, {2*\y}) node[right, fill=white] {\y};
    \foreach \y in {0.5, 1} 
        \draw (0.1, {2*\y}) -- (-0.1, {2*\y}) node[left, fill=white] {\y};

    \begin{scope}[shift={(-5, 1.7)}] 
        \node[draw, fill=white, inner sep=5pt, anchor=west, scale=0.8] at (0, 0) { \begin{minipage}{2.3cm} 
            \begin{tikzpicture}[baseline]
                \draw[blue, thick] (0, 0) -- (0.6, 0); 
                \node[anchor=west] at (0.8, 0) {$\textrm{tanh}(x)$}; 
            \end{tikzpicture}
    
            \begin{tikzpicture}[baseline]
                \draw[red, thick] (0, 0) -- (0.6, 0); 
                \node[anchor=west] at (0.8, 0) {$\hat{\phi}_{\textrm{tanh}}(x)$}; 
            \end{tikzpicture}
        \end{minipage}
        };

    \end{scope}
\end{tikzpicture}
}
      \caption{}
      \label{subfig:tanh_approximation}
    \end{subfigure}
    \caption{Comparison of \subref{subfig:softplus_approximation} $\text{softplus}(x)$ and the semialgebraic function $\hat{\phi}_{\textrm{sp}}(x)$ over the interval $[-5, 5]$ for $c_{\textrm{sp}} = \ln(2)^2 \approx 0.480$ and \subref{subfig:tanh_approximation} $\text{tanh}(x)$ and the semialgebraic function $\hat{\phi}_{\textrm{tanh}}(x)$ over the interval $[-5, 5]$ for $c_{\textrm{tanh}} = 1.171$.}
    \centering
    \label{fig:ApproximateActivationFunctions}
\end{figure*}
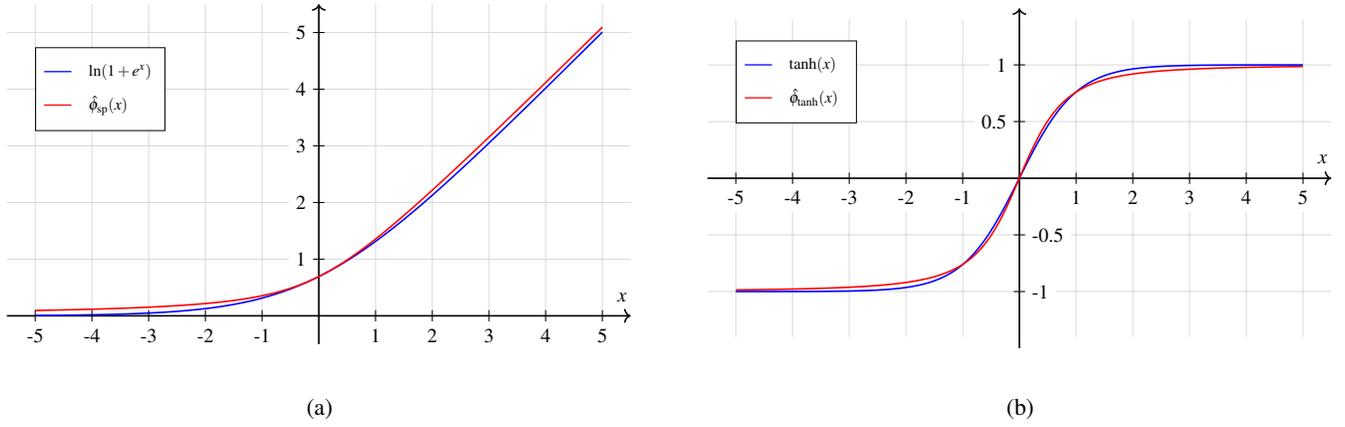

\subsection{Novel Semialgebraic Activation Functions}
\label{sec:NewActivationFunctions}
The review presented in \cref{sec:SoAReview_SAModeling} highlights how the fidelity of the semialgebraic-set-based system model $(\mathbf{K}_\varphi, \mathbf{K}_L)$
is determined by the accuracy with which the graphs of the activation functions $\phi_i$ can be described. Some common activation functions are transcendental and utilize approximate semialgebraic-set-based descriptions of their graph, as illustrated in \cref{exmp:SAset_SingleSigmoid}.
The resulting system models fail to uniquely describe how the states of the closed-loop system change over time, and as a result any function $V_\zeta$ satisfying constraints \cref{eq:LyapunovDecreaseCond} or \cref{eq:LocalLyapunovDecreaseCond} simultaneously decreases in value along \textit{all} model-admissible state trajectories, thereby introducing conservatism in the stability verification procedure.

As an alternative to the inherently approximative approaches currently used to describe transcendental activation functions, we introduce novel activation functions with an exact semialgebraic description that can replace these functions. 
Consider the following alternative to the $\textrm{softplus}$ activation function $\phi_{\textrm{sp}}(x) = \ln(1+e^{x})$,
\begin{equation}
    \label{eq:softplusSAapprox}
    \hat{\phi}_{\textrm{sp}}(x) = \frac{x}{2} + \sqrt{c_{\textrm{sp}} + \Big(\frac{x}{2}\Big)^2} 
\end{equation}
with $c_{\textrm{sp}} > 0$, shown in \cref{subfig:softplus_approximation} for $c_{\textrm{sp}} = \ln(2)^2$ . The function $\hat{\phi}_{\textrm{sp}}$ shares several key properties with $\textrm{softplus}$: both are non-negative, smooth, monotonically increasing, unbounded functions with asymptotes at $0$ and $x$. However, contrary to $\textrm{softplus}$, the graph of $\hat{\phi}_{\textrm{sp}}$ is exactly described by the semialgebraic set 
\begin{equation}
    \label{eq:SAset_SAsoftplus}
    \mathbf{K}_{\hat{\phi}_{\textrm{sp}}} = \left\{ (x, \hat{\phi}_{\textrm{sp}}) \, \left\vert \:
    \begin{alignedat}{1}
         \span \hat{\phi}_{\textrm{sp}} - t_{\textrm{sp}}(x;a) \geq 0 \\
        \span \hat{\phi}_{\textrm{sp}} \big( \hat{\phi}_{\textrm{sp}} - x\big) - c_{\textrm{sp}} = 0
    \end{alignedat}
    \right. \right\}, \ \forall a \in \realsN{},
\end{equation}
with $t_{\textrm{sp}}(x;a) = \od{\hat{\phi}_{\textrm{sp}}(x)}{x}\big\rvert_{x=a}(x-a) + \hat{\phi}_{\textrm{sp}}(a)$ defining the tangent line to $\hat{\phi}_{\textrm{sp}}$ at $x = a$. Likewise, consider the following alternative to the $\textrm{tanh}$ activation function $\phi_{\tanh}$,
\begin{equation}
    \label{eq:tanhSAapprox}
    \hat{\phi}_{\textrm{tanh}}(x) = \frac{c_{\textrm{tanh}}x}{\sqrt{1 + (c_{\textrm{tanh}}x)^2}},
\end{equation}
with $c_{\textrm{tanh}} > 0$, shown in \cref{subfig:tanh_approximation} for $c_{\textrm{tanh}} = 1.171$. Both $\hat{\phi}_{\textrm{tanh}}$ and $\textrm{tanh}$ are odd, smooth, monotonically increasing, bounded functions with asymptotes at $-1$ and $1$.
For all $x \in \realsN{}$, using one lifting variable $\lambda_{\hat{\phi}_{\textrm{tanh}}}(x) = (c_{\mathrm{tanh}}^{-2}
+ x^2)^{0.5}$, the graph of $\hat{\phi}_{\textrm{tanh}}$ is exactly described 
by the semialgebraic set
\begin{equation}
    \label{eq:SAset_SAtanh}
    \mathbf{K}_{\hat{\phi}_{\textrm{tanh}}} =  \left\{ \Bigg(x, \begin{bmatrix} \lambda_{\hat{\phi}_{\textrm{tanh}}} \\ \hat{\phi}_{\textrm{tanh}} \end{bmatrix} \Bigg) \, \left\vert \ 
        \begin{alignedat}{1}
            \span \lambda_{\hat{\phi}_{\textrm{tanh}}}^2 - x^2 - {\textstyle\frac{1}{c_{\textrm{tanh}}^2}} = 0  
            \\
            \span \lambda_{\hat{\phi}_{\textrm{tanh}}} - t_{\textrm{tanh}}(x; a) \geq 0 \\ 
             \span \hat{\phi}_{\textrm{tanh}}\lambda_{\hat{\phi}_{\textrm{tanh}}} - x = 0
        \end{alignedat}
        \right. \right\}, \ \forall a \in \realsN{},
\end{equation}
with $t_{\textrm{tanh}}(x;a) = \od{\lambda_{\hat{\phi}_{\textrm{tanh}}}(x)}{x}\Big\rvert_{x=a}(x-a) + \lambda_{\hat{\phi}_{\textrm{tanh}}}(a)$ defining the tangent line to $\lambda_{\hat{\phi}_{\textrm{tanh}}}$ at $x = a$.

Consequently, \acrshortpl{NNC} synthesized using $\hat{\phi}_{\textrm{sp}}$ and $\hat{\phi}_{\textrm{tanh}}$ instead of $\textrm{softplus}$ and $\tanh$, respectively, form a highly expressive class of controllers whose stability properties can be analyzed with minimal conservatism using the state-of-the-art stability verification procedure reviewed in \cref{sec:SoAReview}.

In addition to the direct use of these new activation functions in newly synthesized \acrshortpl{NNC}, the stability properties of \acrshortpl{NNC} synthesized using $\text{softplus}$ and $\text{tanh}$ activation functions may be analyzed with reduced conservatism by utilizing approximate semialgebraic descriptions derived from \cref{eq:softplusSAapprox,eq:tanhSAapprox}, respectively. A closer examination of the approximative use of these functions will be the subject of a future publication.

\begin{rem}
    Variants of these functions can also be used to approximate other transcendental activation functions, e.g. the logistic function or $\text{arctan}$.
\end{rem}

\subsection{Semialgebraic Modeling of RENs}
\label{sec:NewNNArchitectures}
Previous works utilizing the optimization-based stability verification reviewed in \cref{sec:SoAReview} restricted memoryless \acrshort{NNC} $\varphi$ to be a deep, feedforward neural network. In this section, this assumption is relaxed and compatibility of \acrshortpl{NNC} from the more general class of \acrshortpl{REN} with the semialgebraic-set-based modeling and \acrshort{SDP} formulation procedures of \cref{sec:SoAReview} is proven, which enlarges the class of \acrshort{NNC} $\varphi$ that can be analyzed using the stability verification procedure reviewed in \cref{sec:SoAReview}.

\begin{figure*}[t]
    \centering
    \begin{subfigure}{0.35\textwidth} 
        \centering
\begingroup%
  \makeatletter%
  \providecommand\color[2][]{%
    \errmessage{(Inkscape) Color is used for the text in Inkscape, but the package 'color.sty' is not loaded}%
    \renewcommand\color[2][]{}%
  }%
  \providecommand\transparent[1]{%
    \errmessage{(Inkscape) Transparency is used (non-zero) for the text in Inkscape, but the package 'transparent.sty' is not loaded}%
    \renewcommand\transparent[1]{}%
  }%
  \providecommand\rotatebox[2]{#2}%
  \newcommand*\fsize{\dimexpr\f@size pt\relax}%
  \newcommand*\lineheight[1]{\fontsize{\fsize}{#1\fsize}\selectfont}%
  \ifx\svgwidth\undefined%
    \setlength{\unitlength}{195.70561451bp}%
    \ifx\svgscale\undefined%
      \relax%
    \else%
      \setlength{\unitlength}{\unitlength * \real{\svgscale}}%
    \fi%
  \else%
    \setlength{\unitlength}{\svgwidth}%
  \fi%
  \global\let\svgwidth\undefined%
  \global\let\svgscale\undefined%
  \makeatother%
  \begin{picture}(1,0.63327518)%
    \lineheight{1}%
    \setlength\tabcolsep{0pt}%
    \put(0,0){\includegraphics[width=\unitlength,page=1]{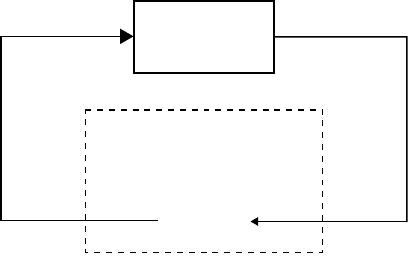}}%
    \put(0.38344241,0.53285429){\makebox(0,0)[lt]{\lineheight{1.25}\smash{\begin{tabular}[t]{l}$x^+ = f(x,u)$\end{tabular}}}}%
    \put(0.05577581,0.12232683){\makebox(0,0)[lt]{\lineheight{1.25}\smash{\begin{tabular}[t]{l}$\varphi(x)$\end{tabular}}}}%
    \put(0.7519792,0.56566617){\makebox(0,0)[lt]{\lineheight{1.25}\smash{\begin{tabular}[t]{l}$x$\end{tabular}}}}%
    \put(0,0){\includegraphics[width=\unitlength,page=2]{RENTransformation_Before.pdf}}%
    \put(0.48412852,0.24785409){\makebox(0,0)[lt]{\lineheight{1.25}\smash{\begin{tabular}[t]{l}$\phi$\end{tabular}}}}%
    \put(0.47730358,0.10231663){\makebox(0,0)[lt]{\lineheight{1.25}\smash{\begin{tabular}[t]{l}{\large $G_\varphi$}\end{tabular}}}}%
    \put(0.23367888,0.19787366){\makebox(0,0)[lt]{\lineheight{1.25}\smash{\begin{tabular}[t]{l}$v$\end{tabular}}}}%
    \put(0.71374434,0.19787366){\makebox(0,0)[lt]{\lineheight{1.25}\smash{\begin{tabular}[t]{l}$\lambda$\end{tabular}}}}%
  \end{picture}%
\endgroup%

        \caption{}
        \label{subfig:PreRENTransformation}
    \end{subfigure}
    \hspace{0.15\textwidth}
    \begin{subfigure}{0.35\textwidth} 
      \centering
\begingroup%
  \makeatletter%
  \providecommand\color[2][]{%
    \errmessage{(Inkscape) Color is used for the text in Inkscape, but the package 'color.sty' is not loaded}%
    \renewcommand\color[2][]{}%
  }%
  \providecommand\transparent[1]{%
    \errmessage{(Inkscape) Transparency is used (non-zero) for the text in Inkscape, but the package 'transparent.sty' is not loaded}%
    \renewcommand\transparent[1]{}%
  }%
  \providecommand\rotatebox[2]{#2}%
  \newcommand*\fsize{\dimexpr\f@size pt\relax}%
  \newcommand*\lineheight[1]{\fontsize{\fsize}{#1\fsize}\selectfont}%
  \ifx\svgwidth\undefined%
    \setlength{\unitlength}{195.70527435bp}%
    \ifx\svgscale\undefined%
      \relax%
    \else%
      \setlength{\unitlength}{\unitlength * \real{\svgscale}}%
    \fi%
  \else%
    \setlength{\unitlength}{\svgwidth}%
  \fi%
  \global\let\svgwidth\undefined%
  \global\let\svgscale\undefined%
  \makeatother%
  \begin{picture}(1,0.63338911)%
    \lineheight{1}%
    \setlength\tabcolsep{0pt}%
    \put(0,0){\includegraphics[width=\unitlength,page=1]{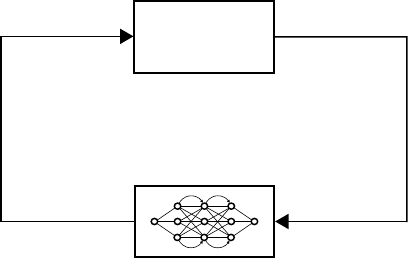}}%
    \put(0.36307978,0.5327426){\makebox(0,0)[lt]{\lineheight{1.25}\smash{\begin{tabular}[t]{l}$\tilde{x}^+ = \tilde{f}(\tilde{x},\tilde{\varphi})$\end{tabular}}}}%
    \put(0.05376934,0.12024305){\makebox(0,0)[lt]{\lineheight{1.25}\smash{\begin{tabular}[t]{l}$\tilde{\varphi}(\tilde{x})$\end{tabular}}}}%
    \put(0.75198121,0.56578001){\makebox(0,0)[lt]{\lineheight{1.25}\smash{\begin{tabular}[t]{l}$\tilde{x}$\end{tabular}}}}%
  \end{picture}%
\endgroup%

      \caption{}
      \label{subfig:PostRENTransformation}
    \end{subfigure}
    \caption{Block diagrams of \subref{subfig:PreRENTransformation} an open-loop system $x^+ = f(x,u)$ in closed-loop with a \acrshort{REN}-based \acrshort{NNC} shown via fractional transformation in the dotted box, and \subref{subfig:PostRENTransformation} the equivalent system obtained by augmenting the state of and input to the dynamical system with the \acrshort{REN}'s internal state variable $x_\varphi$ and hidden variable $\lambda$, respectively.}
    \centering
    \label{fig:RENTransformation}
\end{figure*}

\subsubsection{REN Architecture}
In their most general form, \acrshortpl{REN} are defined by a feedback interconnection consisting of a linear system $G_{\varphi}$ and a memoryless nonlinear operator $\phi$ \cite{mybibfile:Revay2024}, as shown in the dashed box in \cref{subfig:PreRENTransformation}. 
Letting the inputs, outputs and internal state variable associated with linear system $G_{\varphi}$ be denoted by $\bigl[\!\begin{smallmatrix} \lambda \\ x \end{smallmatrix}\!\bigr]$, $\bigl[\!\begin{smallmatrix} v \\ \varphi \end{smallmatrix}\!\bigr]$ and $x_{\varphi}$, respectively, a general \acrshort{REN} is described by
\begin{subequations}
    \label{eq:RENdescription}
    \begin{gather}
        \label{eq:RENSSdescription}
        \begin{bNiceArray}{c}
            x_{\varphi}^+ \\ 
            v \\
            \varphi            
        \end{bNiceArray}
        =
        \begin{bNiceArray}{c|c@{\hskip 5pt}c}
            A & B_{1} & B_{2} \\[1pt] 
            C_{1} & D_{11} & D_{12} \\
            C_{2} & D_{21} & D_{22} 
            \CodeAfter
            \tikz \draw [transform canvas={yshift=1pt}, shorten > = 0.35em, shorten < = 0.35em](2-|1) -- (2-|last) {};
        \end{bNiceArray}
        \begin{bNiceArray}{c}
            x_{\varphi} \\ 
            \lambda \\
            x            
        \end{bNiceArray}
        +
        \begin{bNiceArray}{c}
            b_{x_{\varphi}} \\ 
            b_{v} \\
            b_{\varphi}
        \end{bNiceArray},
        \\
        \label{eq:RENActivationFunction}
        \lambda_i = \phi_i(v_i) \quad \forall i \in [n_\lambda],
    \end{gather}
\end{subequations}
where $n_\lambda$ represents the number of neurons in the \acrshort{REN}. 

The \acrshort{REN} model class encapsulates various standard neural network architectures. Consider the set of \acrshortpl{REN} defined by $b_{x_{\varphi}} = 0$, $A = B_1 = B_2 = C_1 = C_2 = 0$,
\begin{gather}
    \begin{alignedat}{2}
        D_{11} &=
        {
        \setlength{\arraycolsep}{3pt}
        \begin{bNiceArray}{@{}ccccc@{}}[margin,columns-width=1.2em]
            0 		& \mathrlap{\Cdots}	& \Cdots & 	\Cdots \;	&  0  \\
            W_2 	& 0 		& \Cdots &  \Cdots	&  0      \\
            0 		& \Ddots 	& \Ddots &  	&  \Vdots      \\
            \Vdots  & \Ddots  	& \Ddots &  \Ddots  & \Vdots \\
            0 		& \Cdots 	&  0 	 & W_{\ell} 	& 0
        \end{bNiceArray}
        }, \ \ 
        & D_{12} &= 
        {
        \setlength{\arraycolsep}{0.0em}
        \begin{bNiceArray}{@{}c@{}}[margin,columns-width=0.5pt]
            W_1 \\
            0   \\
            \Vdots 	\\
            \Vdots   \\
            0 
        \end{bNiceArray}
        },
        \\
        D_{21} &= 
        {
        \setlength{\arraycolsep}{3pt}
        \begin{bNiceArray}{@{}ccccc@{}}[margin,columns-width=1.2em]
            0 & \Cdots & \Cdots & 0 & W_{\ell+1}
        \end{bNiceArray}
        }, \ \
        & D_{22} &= \;
        {
        \setlength{\arraycolsep}{0.1em}    
        \begin{bNiceArray}{@{}c@{}}[margin,columns-width=0.5pt]
            \hspace{0.3em} 0   
        \end{bNiceArray}
        }, 
        \\
        b_v\transpose &= 
        {
        \setlength{\arraycolsep}{3pt}
        \begin{bNiceArray}{@{}ccccc@{}}[margin,columns-width=1.15em]
            b_1\transpose & \Cdots & \Cdots & b_{\ell-1}\transpose &   b_{\ell}\transpose 
        \end{bNiceArray}
        }, \ \ 
        & b_\varphi\transpose &= \;
        b_{\ell + 1}\transpose,
    \end{alignedat}
\end{gather}
and recall \cref{eq:FFNNdefinition}, \cref{eq:FFNNLifted}:
it follows that the general \acrshort{REN} definition contains all deep, feedforward neural networks. In fact, the general description of \cref{eq:RENdescription} also contains recurrent neural networks such as \acrshort{LSTM} networks, convolutional networks, and more \cite{mybibfile:Revay2024}. 

As a result of their versatility, \acrshortpl{REN} may not always define a well-posed mapping from $\realsN{n_x} \times \realsN{n_{x_\varphi}}$ to $\realsN{n_u}$. As shown in \cite{mybibfile:Revay2020_EquilibriumNet}, sufficient conditions to guarantee the existence and uniqueness of $[(x_{\varphi}^+)\transpose, \, v\transpose,\, \lambda\transpose,\,\varphi\transpose]\transpose$ for all values of $\bigl[\!\begin{smallmatrix} x \\ x_\varphi \end{smallmatrix}\!\bigr]$ are:
\begin{enumerate}
    \item All scalar functions comprising $\phi$, i.e. $\phi_i$ for $i \in [n_{\lambda}]$, are monotone and slope-restricted in $[0, 1]$.
    \item There exists a positive definite, diagonal matrix $\Lambda$ such that
\begin{equation}
    \label{eq:WellPosedREN}
    2\Lambda - \Lambda D_{11} - D_{11}\transpose\Lambda \succ 0.
\end{equation}
\end{enumerate}

\subsubsection{Formulation of Equivalent System}
\label{sec:REN_EquivalentSystem}
Consider now an open-loop system with polynomial dynamics, $f(x,u)\colon \realsN{n_x} \times \realsN{n_u} \mapsto \realsN{n_x}$, controlled by a well-posed, \acrshort{REN}-based \acrshort{NNC}, $u(x) = \varphi(x, x_{\varphi})\colon \realsN{n_x} \times \realsN{n_{x_\varphi}} \mapsto \realsN{n_u}$ described by \cref{eq:RENdescription}. This closed-loop system, described by \cref{eq:RENdescription} and
\begin{equation}
    \label{eq:GeneralRENClosedLoopSystem}
    x^+ = f\big(x, \varphi(x, x_{\varphi})\big)
\end{equation}
is shown in \cref{subfig:PreRENTransformation}. Without loss of generality this system is assumed to have an equilibrium point at the origin, i.e. $f\big(0, \varphi(0, 0)\big) = 0$ and $B_1\lambda(x = 0, x_\varphi = 0) + b_{x_{\varphi}} = 0$. 

To analyze the stability properties of closed-loop system \cref{eq:GeneralRENClosedLoopSystem}, we formulate an equivalent closed-loop system compatible with the state-of-the-art stability verification procedure of \cref{sec:SoAReview}. 
Consider a closed-loop system with augmented state $\tilde{x} = \bigl[\!\begin{smallmatrix} x \\ x_\varphi \end{smallmatrix}\!\bigr]$, augmented control input $\tilde{\varphi} = \bigl[\!\begin{smallmatrix} \lambda \\ \varphi \end{smallmatrix}\!\bigr]$ generated by a memoryless \acrshort{NNC} $\tilde{\varphi} \colon \realsN{n_x + n_{x_\varphi}} \mapsto \realsN{n_\lambda + n_u}$ and augmented dynamics $\tilde{f} \colon \realsN{n_x + n_{x_{\varphi}}} \times \realsN{n_\lambda + n_{u}} \mapsto \realsN{n_x + n_{x_\varphi}}$.
This closed-loop system, shown in \cref{subfig:PostRENTransformation}, is governed by
\begin{subequations}
    \label{eq:EquivalentRENSystem}
    \begin{alignat}{1}
        \label{eq:EquivalentRENSystem_Dynamics}
        \tilde{x}^+ &= \tilde{f}(\tilde{x}, \tilde{\varphi}(\tilde{x})\big) = 
        \begin{bmatrix}
            f\big(x, \varphi(\tilde{x})\big) \\
            \begin{bmatrix} B_2 & A \end{bmatrix} \tilde{x} + B_1\lambda(\tilde{x}) + b_{x_{\varphi}}
        \end{bmatrix},
        \\
        \label{eq:EquivalentRENSystem_NNC}
        \ \tilde{\varphi}(\tilde{x}) &
        = 
        \left\{
        \begin{aligned}
            & \begin{bmatrix} 
                \lambda \\
                \varphi
            \end{bmatrix} 
            = 
            \begin{bmatrix} 
            I & 0 \\
            D_{21} & \begin{bmatrix} D_{22} & C_2 \end{bmatrix}
            \end{bmatrix}
            \begin{bmatrix}
                \lambda \\ \tilde{x}
            \end{bmatrix} +
            \begin{bmatrix}
                0 \\
                b_{\varphi}
            \end{bmatrix},
            \\
            \span \textrm{with} \\
            & \ \lambda 
            = \phi(v) = 
            \phi\big(D_{11}\lambda + \begin{bmatrix} D_{12} & C_1\end{bmatrix} \tilde{x} + b_v \big),
        \end{aligned}
        \right.
    \end{alignat}
\end{subequations}
making it equivalent to closed-loop system \cref{eq:GeneralRENClosedLoopSystem}. 
As a result, the stability properties of closed-loop system \cref{eq:GeneralRENClosedLoopSystem} are examined via a stability analysis of closed-loop system \cref{eq:EquivalentRENSystem}.

\subsubsection{Stability Verification of Equivalent System}
The semialgebraic-set-based modeling procedure of \cref{sec:SoAReview_SAModeling} can be applied to obtain 
a model of \acrshort{NNC} $\tilde{\varphi}$.

\begin{thm}[Semialgebraic Modeling of REN \protect\texorpdfstring{\ensuremath{\tilde{\varphi}}}{phi tilde}]
    \label{thm:REN_SAModeling}
    Consider a memoryless, well-posed \acrshort{REN} $\tilde{\varphi}$ defined by \cref{eq:EquivalentRENSystem_NNC} and comprised of activation functions $\phi_i$ that permit a possibly inexact semialgebraic description of their graph. The graph of $\realsN{n_x + n_{x_\varphi}} \ni \tilde{x}
    \mapsto \tilde{\varphi}(\tilde{x})$ permits a semialgebraic description.
\end{thm}
\begin{proof}
    By assumption, there exist possibly inexact semialgebraic descriptions of each activation function, i.e. for all $i \in [n_\lambda]$ it holds $\graph(\phi_i) = \big\{ \big(v_i, \phi_i(v_i)\big) \mid v_i \in \realsN{} \big\} \subseteq \mathbf{K}_{\phi_i}$ defined by
    \begin{equation}
        \label{eq:REN_Semialgebraicity_ActivationFunctions}
        \mathbf{K}_{\phi_i} = \big\{ (v_i, \lambda_i) \mid g^{\varphi_i}(v_i, \lambda_i) \geq 0, \, h^{\varphi_i}(v_i, \lambda_i) = 0 \big\}.
    \end{equation}
    For any well-posed \acrshort{REN} $\tilde{\varphi}(\tilde{x})$ exists and is unique. Thus, as any pair $\big(\tilde{x}, \tilde{\varphi}(\tilde{x})\big)$ satisfies \cref{eq:EquivalentRENSystem_NNC}, it follows by construction that $\graph(\tilde{\varphi}) \subseteq \mathbf{K}_{\tilde{\varphi}} = \big\{ \big( \tilde{x}, \bigl[\!\begin{smallmatrix} \lambda \\ \varphi \end{smallmatrix}\!\bigr] \big) \big\vert \, g^{\tilde{\varphi}}(\tilde{x}, \lambda, \varphi) \geq 0, \ h^{\tilde{\varphi}}(\tilde{x}, \lambda, \varphi) = 0 \big\}$ with $g^{\tilde{\varphi}}(\tilde{x}, \lambda, \varphi) =$
    \begin{gather}
        \begin{bmatrix}
            g^{\varphi_1}(\bigl[D_{11}\bigr]_{1} \lambda + \bigl[ D_{12} \ C_1  \bigr]_{1} \tilde{x} + b_{v,1}, \, \lambda_1) \\
            \vdots 
            \\
            g^{\varphi_{n_\lambda}}(\bigl[D_{11}\bigr]_{n_\lambda} \lambda + \bigl[ D_{12} \ C_1  \bigr]_{n_\lambda} \tilde{x} + b_{v,n_\lambda}, \, \lambda_{n_\lambda})
        \end{bmatrix},
        \intertext{and $h^{\tilde{\varphi}}(\tilde{x}, \lambda, \varphi) =$}
        \begin{bmatrix}
            h^{\varphi_1}(\bigl[D_{11}\bigr]_{1} \lambda + \bigl[ D_{12} \ C_1  \bigr]_{1} \tilde{x} + b_{v,1}, \lambda_1) \\
            \vdots 
            \\
            h^{\varphi_{n_\lambda}}(\bigl[D_{11}\bigr]_{n_\lambda} \lambda + \bigl[ D_{12} \ C_1  \bigr]_{n_\lambda} \tilde{x} + b_{v,n_\lambda}, \lambda_{n_\lambda})
            \\
            \varphi - D_{21}\lambda - \bigl[ D_{22} \ C_2  \bigr]^{\phantom{|}}  \!\tilde{x} - b_{\varphi}
        \end{bmatrix}.
    \end{gather}
\end{proof}

The proof of \cref{thm:REN_SAModeling} outlines how the semialgebraic set $\mathbf{K}_{\tilde{\varphi}}$ 
can be constructed. In case the activation functions $\phi_i$ for $i \in [n_\lambda]$ comprising
$\tilde{\varphi}$ permit an exact semialgebraic description, e.g. all $\phi_i$ are $\textrm{ReLU}$ or any of the novel semialgebraic activation functions presented in \cref{sec:NewActivationFunctions}, the semialgebraic description $\mathbf{K}_{\tilde{\varphi}}$ is exact.

\begin{cor}
    \label{cor:REN_SAModeling}
    Consider the memoryless, well-posed \acrshort{REN} $\tilde{\varphi}$ defined by \cref{eq:EquivalentRENSystem_NNC} and comprised of activation functions $\phi_i$ that permit an exact semialgebraic description of their graph.
    The graph of $\realsN{n_x + n_{x_\varphi}} \ni \tilde{x} 
    \mapsto \tilde{\varphi}(\tilde{x})$ is a semialgebraic set that can be represented exactly.
\end{cor}
\begin{proof}
    By \cref{thm:REN_SAModeling} it follows $\graph(\tilde{\varphi}) \subseteq \mathbf{K}_{\tilde{\varphi}}$. If $\graph(\phi_i) = \mathbf{K}_{\phi_i}$ for all $i \in [n_\lambda]$, it follows that any pair $\big( \tilde{x}, \bigl[\!\begin{smallmatrix} \lambda \\ \varphi \end{smallmatrix}\!\bigr] \big) \in \mathbf{K}_{\tilde{\varphi}}$ is a solution to \cref{eq:EquivalentRENSystem_NNC}, i.e. $\graph(\tilde{\varphi}) \supseteq \mathbf{K}_{\tilde{\varphi}}$, thereby proving $\graph(\tilde{\varphi}) = \mathbf{K}_{\tilde{\varphi}}$.
\end{proof}

In addition to the semialgebraic-set-based model $\mathbf{K}_{\tilde{\varphi}}$ of $\tilde{\varphi}$, a semialgebraic-set-based model $\mathbf{K}_{\tilde{L}}$ of $\tilde{L} = \tilde{\varphi} \circ \tilde{f} \circ (\textrm{id}, \tilde{\varphi})$ can also be constructed as \cref{eq:EquivalentRENSystem_Dynamics} remains polynomial. Thus, a semialgebraic-set-based system model $\big(\mathbf{K}_{\tilde{\varphi}}, \mathbf{K}_{\tilde{L}}\big)$ of closed-loop system \cref{eq:EquivalentRENSystem} can be constructed that is directly compatible with the \acrshort{SDP} formulation procedure reviewed in \cref{sec:SoAReview_SDPFormulation}, enabling the stability properties of closed-loop systems \cref{eq:GeneralRENClosedLoopSystem,eq:EquivalentRENSystem} to be analyzed. 

\begin{rem}
    The construction of system model $\big( \mathbf{K}_{\tilde{\varphi}}, \mathbf{K}_{\tilde{L}} \big)$ can easily be adapted, e.g. to obtain a local system model, to obtain a system model utilizing semialgebraic sets $\mathbf{K}_{\phi_i}$ with lifting variables as in \cref{eq:SAset_SAtanh}, etc.
\end{rem}

Assuming continuity of $\phi_i$ for all $i \in [n_\lambda]$, the conditions guaranteeing well-posedness of a \acrshort{REN} are also sufficient to ensure that $\lambda$ and $\varphi$ are continuous functions in $\tilde{x}$, allowing the general form of \cref{eq:LyapunovParam} to be used to parameterize functions $V \colon \realsN{n_x} \mapsto \realsN{}_{\geq 0}$. See Appendix \labelcref{sec:AppendixRENContinuity} for a proof.

In conclusion, the modeling improvements presented in this section 
allow an exact semialgebraic-set-based system model to be constructed for a larger class of \acrshort{NNC}, 
including \acrshort{NNC} utilizing the newly introduced activation functions $\hat{\phi}_{\textrm{sp}}$ and $\hat{\phi}_{\textrm{tanh}}$, well-posed, \acrshort{REN}-based \acrshort{NNC}, and combinations thereof. This allows the stability properties of systems controlled by this larger class of \acrshort{NNC} to be analyzed with minimal conservatism using the state-of-the-art procedure reviewed in \cref{sec:SoAReview}.

\section{Improved SDP formulations for local stability}
\label{sec:StabilityVerificationContributions}
In this section we present our contributions to the \acrshort{SDP} formulation step as reviewed in \cref{sec:SoAReview_SDPFormulation}. We specifically address the search for local Lyapunov functions by presenting improved \acrshort{SDP} formulations searching for this type of local stability certificate. Such formulations are of great practical importance as a globally stabilizing \acrshort{NNC} is often not required and may not even exist due to state and/or input constraints. 

The current state-of-the-art \acrshort{SDP} formulation searching for local Lyapunov functions reviewed in \cref{sec:SOAReview_SDPFormulation_LocalStability} suffers from two important drawbacks:
\begin{enumerate}
    \item The parameterization \cref{eq:LyapunovParam} 
    in \acrshort{SDP} \cref{eq:SDPFormulationLocalAsymptoticStability} is restricted to a small class of known candidate Lyapunov functions $V$ to ensure any solution defines a local Lyapunov function 
    for closed-loop system \cref{eq:GeneralClosedLoopSystem}.
    \item The inner estimate of the \acrshort{ROA} provided via the current stability verification procedure is not optimized directly.
\end{enumerate}
These drawbacks are addressed in \cref{sec:RefinedParameterizationOfSolutionSpace,sec:ImprovedSDPFormulations} by presenting a class of local Lyapunov functions compatible with \acrshort{SDP} \cref{eq:SDPFormulationLocalAsymptoticStability} richer than any previously reported and the development of a novel \acrshort{SDP} formulation directly optimizing an inner estimate of the \acrshort{ROA} for a general closed-loop system \cref{eq:GeneralClosedLoopSystem}, respectively.

\subsection{Refined Parameterization of Solution Space}
\label{sec:RefinedParameterizationOfSolutionSpace}
Given a semialgebraic-set-based system model $\big( \mathbf{K}_\varphi, \mathbf{K}_L\big)$, the most general parameterization of functions $V \colon \realsN{n_x} \mapsto \realsN{}_{\geq 0}$ is defined by \cref{eq:LyapunovParam}.
Whilst this general form encompasses a rich class of candidate Lyapunov functions, it also parameterizes other non-negative functions. This does not present a problem when solving \acrshort{SDP} \cref{eq:SDPFormulationGlobalAsymptoticStability}, as it can be shown that all solutions to this \acrshort{SDP} are Lyapunov functions. However, the same property does not hold for \acrshort{SDP} \cref{eq:SDPFormulationLocalAsymptoticStability}, which forms part of the current state-of-the-art \acrshort{SDP} formulation that is used to find a local Lyapunov function for closed-loop system \cref{eq:GeneralClosedLoopSystem}. This is illustrated in the following example.

\begin{exmp}
    \label{exmp:CounterexampleLAS}
    Consider an analysis of the local stability properties of the closed-loop system $x^+ = 2x$ over the set $\mathcal{Q} = \big\{ x \in \realsN{} \mid x^2 \leq \frac{1}{4} \big\}$. Despite the closed-loop system being trivially unstable, it will be shown that solutions to \acrshort{SDP} \labelcref{eq:SDPFormulationLocalAsymptoticStability} exist.

    One such solution can be found by considering the \acrshort{SOS} function $V(x) = \frac{1}{5}(x-2)^2(x+2)^2$, which is contained in the parameterization of \cref{eq:LyapunovParam}.
    As a result of the local maximum of $V$ at the equilibrium point, constraint \cref{eq:LocalLyapunovDecreaseCond} is satisfied for this function $V$ over the set $\mathcal{Q}$. This can be proven by setting $\sigma^{\Delta V} = 0$, $\sigma^{\Delta V}_{\textrm{ineq}} = x^4$ and directly substituting $x^+$ according to its definition. \Cref{eq:LocalLyapunovDecreaseCond} then reads
    \begin{equation}
         V(x) - V(2x) - \|x\|^2  - x^4 \left( \frac{1}{4} - x^2 \right) \geq 0,
    \label{eq:CounterexampleLASLyapunovDecrease}
    \end{equation}
    which is satisfied for this choice of $V$ since the left-hand side of \cref{eq:CounterexampleLASLyapunovDecrease} possesses the \acrshort{SOS} decomposition
    \begin{equation}
        \left(\sqrt{\frac{95}{25}}x -\sqrt{\frac{4655}{5776}}x^3 \right)^2 + \left(\frac{1}{2}x^2\right)^2 + \left(\sqrt{\frac{1121}{5776}}x^3\right)^2.
    \end{equation}
    Thus, this example illustrates that a solution to \acrshort{SDP} \labelcref{eq:SDPFormulationLocalAsymptoticStability} is not sufficient to guarantee the system is \acrshort{LAS}. 
\end{exmp}

To guarantee that any solution to \acrshort{SDP} \cref{eq:SDPFormulationLocalAsymptoticStability} defines a valid local Lyapunov function, state-of-the-art local stability analyses restrict the general parameterization of \cref{eq:LyapunovParam} to consist of a class of known candidate Lyapunov functions, e.g. $V(x) = x\transpose Px$, $P\succ0$ \cite{mybibfile:Korda2017,mybibfile:Yin2022,mybibfile:Pauli2021,mybibfile:Newton2022}. 

Under the assumption that a semialgebraic-set-based system model $( \mathbf{K}_{\varphi}, \mathbf{K}_L)$ is available, we present a class of explicit candidate Lyapunov functions compatible with \acrshort{SDP} \cref{eq:SDPFormulationLocalAsymptoticStability} that is richer than any previously reported. 

Recall the general parameterization of \cref{eq:LyapunovParam} as 
$V_\zeta = \sigma^V + (\sigma^{V}_{\textrm{ineq}})\transpose g^V$, where $g^V$
consists of unique products generated by the entries of $g^{\varphi}$. A parameterization consisting only of candidate Lyapunov functions is obtained by minimally restricting the form of $\sigma^V$, $\sigma^V_{\textrm{ineq}}$ such that each term comprising $V$ has a minimum at $x=0$. Let 
$\mathcal{I}^{g^\varphi}_{>0}$
be the index set of entries $g^{\varphi}$ that evaluate to a positive value for $x = 0$,
\begin{equation}
    \label{eq:Vparam_Sdef}
    \mathcal{I}^{g^\varphi}_{>0} = \Big\{ i \in [n_{g^\varphi}] \: \big\vert \: g^{\varphi}_i\big(0, \lambda(0), \varphi(0)\big) > 0\Big\},
\end{equation}
define the set of all entries of $g^V$ generated using only entries $\mathcal{I}^{g^\varphi}_{>0}$ of $g^\varphi$ as $\mathcal{G}^V_{>0}$,
\begin{equation}
    \label{eq:Vparam_Ggeq0def}
    \mathcal{G}^V_{>0}  = \bigg\{ \prod_{j \in \mathcal{J}} g^{\varphi}_j \mid \mathcal{J} \subseteq \mathcal{I}^{g^\varphi}_{>0} \bigg\},
\end{equation}
and let
\begin{equation}
    \Delta\zeta = \zeta - 
    {\renewcommand{\arraystretch}{0.9}%
        \begin{bmatrix}
        0 \\
        \lambda(0) \\
        \varphi(0)
        \end{bmatrix}
    }.
\end{equation} 
Following the notation introduced in \cref{eq:SOSDefinition},
let $\nu_{\geq 1}(\Delta \zeta)$ denote a vector of monomial terms generated by entries of $\Delta \zeta$ that are at least of degree one.

Now consider the parameterization of \cref{eq:LyapunovParam} where 
\begin{subequations}
    \label{eq:LyapunovParam_LocalStability}
    \begin{align}
        \sigma^V &= x\transpose P x + \nu_{\geq 1}( \Delta\zeta)\transpose Q \nu_{\geq 1}( \Delta\zeta),
        \\
        \sigma^V_{\textrm{ineq},i} &= 
        \begin{cases}
            \nu_{\geq 1}( \Delta\zeta )\transpose Q_i \nu_{\geq 1}( \Delta\zeta ) 
            & 
            \text{if }
                g^{V}_i \in \mathcal{G}^V_{>0},
            \\
            \nu( \Delta\zeta )\transpose Q_i \nu( \Delta\zeta )
            & \text{otherwise, }
        \end{cases}
    \end{align}
\end{subequations}
with parameters $P \succ 0$ and $Q, \, Q_i \succeq 0$ for all $i$. By construction, for all valid parameters $P, \, Q, \, Q_i$, any continuous function $V$ specified by \cref{eq:LyapunovParam,eq:LyapunovParam_LocalStability} 
is a valid candidate Lyapunov function. 
The following example illustrates the expressiveness of this function class and highlights how this parameterization is well-suited for finding (low-degree) stability certificates when analyzing \acrshort{NNC} $\varphi$, particularly those using $\textrm{ReLU}$ activation functions and/or any of the newly introduced activation functions $\hat\phi_{\textrm{sp}}$, $\hat\phi_{\textrm{tanh}}$.

\begin{exmp}
    Consider an open-loop, polynomial dynamical system $x^+ = f(x,u)$
    with $f \colon \realsN{2} \times \realsN{2} \mapsto \realsN{2}$ subject to the \acrshort{NNC}
    \begin{equation}
        \varphi(x) = \begin{bmatrix}
            \phi_{\textrm{ReLU}}(x_1 - 1) 
            \\
            \hat{\phi}_{\textrm{sp}}(x_2)
        \end{bmatrix}.
    \end{equation}
    By \cref{eq:SAset_SAsoftplus,eq:SingleReLUNetwork}, the vector-valued, polynomial function $g^\varphi(x, \varphi)$ defining the inequalities of the semialgebraic set $\mathbf{K}_\varphi = \graph(\varphi)$ is then given as 
    \begin{equation}
        g^\varphi(x, \varphi) = \begin{bmatrix}
            \varphi_1 \\
            \varphi_1 - x_1 + 1 \\
            \varphi_2 - t_{\textrm{sp}}(x_2; a)
        \end{bmatrix}, \ \forall a \in \realsN{}.
    \end{equation}
    By choosing $a = 0$, from \cref{eq:Vparam_Sdef,eq:Vparam_Ggeq0def} it follows $\mathcal{I}_{>0}^{g^\varphi} = \{2\}$, $\mathcal{G}^V_{>0} = \{ \varphi_1 - x_1 + 1 \}$. Thus, considering candidate Lyapunov functions quadratic in $x$, $\varphi(x)$, by \cref{eq:LyapunovParam,eq:LyapunovParam_LocalStability} we obtain
    \begin{equation}
        \begin{alignedat}{1}
            V(x) &= 
            \begin{bmatrix}
            x \\
            \varphi(x) - \varphi(0)
            \end{bmatrix}\transpose
            \begin{bmatrix}
                P + Q_{11} & Q_{12} \\
                Q_{12}     & Q_{22}
            \end{bmatrix}
            \begin{bmatrix}
                * \\
                *
            \end{bmatrix}
            + Q_1 \varphi_1(x)
            \\ & \quad \ \ 
            + Q_2 \big(\varphi_2(x) - t_{\textrm{sp}}(x_2;0)\big) 
            + Q_3 \varphi_1(x)\big( \varphi_2(x) - t_{\textrm{sp}}(x_2;0) \big) 
            \\ & \quad \ \
            + Q_4 \big( \varphi_2(x) - t_{\textrm{sp}}(x_2;0) \big) \big( \varphi_1(x) - x_1 + 1 \big),
        \end{alignedat}
    \end{equation}
    with $P \succ 0$, $Q = \bigl[ \begin{smallmatrix} Q_{11} & Q_{12} \\ Q_{12} & Q_{22} \end{smallmatrix} \bigr] \succeq 0$, $Q_1, Q_2, Q_3, Q
    _4 \geq 0$. This parameterization can be used in \acrshort{SDP} \cref{eq:SDPFormulationLocalAsymptoticStability} to search for a local Lyapunov function for $x^+ = f\big(x, \varphi(x)\big)$.
\end{exmp}

Thus, the parameterization of \cref{eq:LyapunovParam,eq:LyapunovParam_LocalStability} defines a class of candidate Lyapunov functions $V$ compatible with the state-of-the-art semialgebraic-set-based local stability verification procedure reviewed in \cref{sec:SoAReview}.
Relative to prior formulations \cite{mybibfile:Korda2017,mybibfile:Yin2022,mybibfile:Pauli2021,mybibfile:Newton2022}, this parameterization defines a larger finite-dimensional function space of candidate Lyapunov functions, allowing \acrshort{SDP} \cref{eq:SDPFormulationLocalAsymptoticStability} to search through a larger set of possible solutions whilst still ensuring that any solution defines a valid local Lyapunov function for closed-loop system \cref{eq:GeneralClosedLoopSystem}.

\subsection{Novel SDP Formulation for Local Stability}
\label{sec:ImprovedSDPFormulations}
The current state-of-the-art \acrshort{SDP} formulation used to find a local Lyapunov function for closed-loop system \cref{eq:GeneralClosedLoopSystem} does not directly optimize the size of the resulting \acrshort{ROA} estimate, despite 
this being of great practical importance. We therefore present an alternative formulation of \acrshort{SDP} problems that directly optimize this local stability property. This is achieved by first formulating a nonconvex, augmented optimization problem, any solution to which certifies $\mathcal{Q}$ of \cref{eq:LocalRegionQDefinition} to directly be an estimate of the \acrshort{ROA}. Next, in an attempt to find the largest set $\mathcal{Q}$ for which this augmented optimization problem is feasible, a sequence of \acrshortpl{SDP} is formulated to solve this augmented optimization problem with the aim of iteratively increasing the size of $\mathcal{Q}$.

\subsubsection{Formulation of Augmented Optimization Problem}
First, an augmented optimization problem is formulated whose solutions directly yield an \acrshort{ROA} estimate. This is achieved by allowing the local region of the state space being examined, $\mathcal{Q}$, to be an optimization variable. Specifically, by \cref{eq:LocalRegionQDefinition}, consider $\mathcal{Q}$ to be defined by a function $q$ of the form
\begin{equation}
    \label{eq:RealizationOfSetQ}
    q(x) = q_{\zeta}\big(x, \lambda(x), \varphi(x)\big) = \alpha - \sigma_q\big(x, \lambda(x), \varphi(x)\big)
\end{equation}
where $\alpha$ is a scalar variable and $\sigma_q$ is constrained to be a \acrshort{SOS} polynomial in $\zeta$ equal to zero at $x=0$. This constraint is denoted $\sigma_q \in \mathcal{Q}_0$ and can be realized by restricting the form of $\sigma_q$ similarly to \cref{eq:LyapunovParam_LocalStability}, e.g. $\sigma_q(\zeta) = \nu_{\geq 1}(\Delta\zeta)\transpose Q \nu_{\geq 1}(\Delta\zeta)$, $Q \succeq 0$. 
By \cref{assume:ContinuityOfQinx} it follows $0 \in \textrm{int}(\mathcal{Q})$ for $\alpha > 0$.

Next, the set $\mathcal{Q}$ defined by the optimization variables of \cref{eq:RealizationOfSetQ} is constrained to be invariant. Based on the concept of a discrete-time barrier function \cite{mybibfile:Agrawal2017}, consider the constraint
\begin{multline}
    \label{eq:BarrierCond}
    \| \xi \|^{2k} q_\zeta(\xi_{x^+}, \xi_{\lambda^+}, \xi_{\varphi^+}) \geq 
    p^{\Delta \mathcal{Q}}_{\textrm{eq}}(\xi)\transpose h^L(\xi) + \sigma^{\Delta \mathcal{Q}}(\xi) \\
    + \sigma^{\Delta \mathcal{Q}}_{\textrm{ineq}}(\xi) \transpose
    \begin{bmatrix}
        \mathcal{M}\Big( \left[\!\begin{smallmatrix} q_\zeta(\xi_{x}, \xi_{\lambda}, \xi_{\varphi}) \\ g^L(\xi) \end{smallmatrix}\!\right]^{\phantom{|}}\!\!, 1 \Big)
        \\
        \mathcal{M}\Big( \left[\!\begin{smallmatrix} q_\zeta(\xi_{x}, \xi_{\lambda}, \xi_{\varphi}) \\ g^L(\xi) \end{smallmatrix}\!\right]^{\phantom{|}}\!\!, 2 \Big)
        \\
        \vdots
    \end{bmatrix}, \  \forall \xi \in \realsN{n_\xi}
\end{multline}
with $\sigma^{\Delta \mathcal{Q}}$ a scalar \acrshort{SOS} polynomial, $\sigma^{\Delta \mathcal{Q}}_{\textrm{ineq}}$ a vector of \acrshort{SOS} polynomials, $p^{\Delta \mathcal{Q}}_{\textrm{eq}}$ a vector of arbitrary polynomials and $k \in \integersN{}_{\geq 0}$. By substituting the variables $\xi_x$, $\xi_\lambda$, $\xi_\varphi$, $\xi_x^+$, $\xi_\lambda^+$, $\xi_\varphi^+$ in \cref{eq:BarrierCond} with $x$, $\lambda(x)$, $\varphi(x)$, $x^+(x)$, $\lambda^+(x)$, $\varphi^+(x)$ it follows that satisfaction of \cref{eq:BarrierCond} is a sufficient condition to guarantee $q\big(x^+(x)\big) \geq 0$ for all $x \in \mathcal{Q}$, thereby ensuring invariance of $\mathcal{Q}$. 

Interpreting \cref{eq:BarrierCond} as an \acrshort{SOS} constraint and combining this with previous constraints leads to the augmented optimization problem
\begin{subequations}
    \label{eq:OptimProblemFormulationInvariantLocalAsymptoticStability}
    \begin{alignat}{4}
        &\span\span \text{find:} \, &  \; \Pi = \big\{\sigma^V, \sigma^V_{\textrm{ineq}},  
        \sigma^{\Delta V},  \sigma^{\Delta V}_{\textrm{ineq}}, &  \sigma^{\Delta \mathcal{Q}}, \sigma^{\Delta \mathcal{Q}}_{\textrm{ineq}}, \sigma_q, p^{\Delta V}_{\textrm{eq}}, p^{\Delta \mathcal{Q}}_{\textrm{eq}}, \alpha \big\} \span \nonumber \\
        &\span\span \text{s.t.} \ & \eqref{eq:LyapunovParam}, \, \, \eqref{eq:LocalLyapunovDecreaseCond},  \ &  \eqref{eq:RealizationOfSetQ}, \, \, \eqref{eq:BarrierCond} \span  \\
        &\span\span & \begin{gathered} 
        \sigma^V, \sigma^V_{\textrm{ineq}}, \sigma^{\Delta V},  
        \\
        \sigma^{\Delta V}_{\textrm{ineq}}, \sigma^{\Delta \mathcal{Q}},\sigma^{\Delta \mathcal{Q}}_{\textrm{ineq}},
        \end{gathered} \quad & \text{\acrshort{SOS} polynomials,} \label{eq:InvariantLocalAsymptoticStabilitySOS} \\
        &\span\span & 
        p^{\Delta V}_{\textrm{eq}}, p^{\Delta \mathcal{Q}}_{\textrm{eq}} \quad & \text{arbitrary polynomials,}  \\
        &\span\span & \sigma_q \in & \ \mathcal{Q}_0.   \\
        &\span\span & \alpha > & \ 0.  \label{eq:InvariantLocalAsymptoticStabilityGammaStrictlyPositive}
    \end{alignat}
\end{subequations}
Unlike \acrshort{SDP} \cref{eq:SDPFormulationLocalAsymptoticStability}, the unrestricted parameterization of \cref{eq:LyapunovParam} can be used in this optimization problem without losing the guarantee that any $V$ forming a solution to this problem defines a local Lyapunov function.

\begin{thm}
    \label{thm:ValidityOptimProblemInvariantLocalAsymptoticStability}
    Under \cref{assume:ContinuityOfV,assume:ContinuityOfQinx}, and assuming $f \circ (\textrm{id}, \varphi) \colon \realsN{n_x} \mapsto \realsN{n_x}$ is locally bounded, any solution to augmented optimization problem \cref{eq:OptimProblemFormulationInvariantLocalAsymptoticStability} defines a Lyapunov function in $\mathcal{Q}$ for closed-loop system \cref{eq:GeneralClosedLoopSystem}, thereby certifying that this system is \acrshort{LAS} and that $\mathcal{Q}$ forms an inner estimate of the \acrshort{ROA}.
\end{thm}
\begin{proof} 
    \label{prf:ProofValidityOptimProblemInvariantLocalAsymptoticStability}
    Given any solution to augmented optimization problem \cref{eq:OptimProblemFormulationInvariantLocalAsymptoticStability}, it will be shown that $\tilde{V}(x) = V(x)-V(0)$ constitutes a valid Lyapunov function in $\mathcal{Q}$ for closed-loop system \cref{eq:GeneralClosedLoopSystem} by constructing functions $\alpha_1, \alpha_2 \in \mathcal{K}_{\infty}$, and continuous, positive definite function $\alpha_3$ of \cref{defn:LyapunovFunction}. By a standard Lyapunov proof, the existence of this local Lyapunov function is a certificate that the system is \acrshort{LAS} and that $\mathcal{Q}$ forms an inner estimate of the \acrshort{ROA}.
    
    Let $V$ be defined by a solution to augmented optimization problem \cref{eq:OptimProblemFormulationInvariantLocalAsymptoticStability}. First, it will be shown that any such function $V$ satisfies
    \begin{equation}
        \argmin_{x \in \mathcal{Q}} \  V\left(x\right) = \{ 0 \}.
        \label{eq:MinimumOfV}
    \end{equation}
    This can be seen by considering 
    \begin{align}
        \mathcal{Q}_1 = \{ x \in \mathcal{Q} \mid \|x\|^2 \leq V(0)\}, \\
        \mathcal{Q}_2 = \{ x \in \mathcal{Q} \mid \|x\|^2 > V(0)\},
    \end{align}
    such that $\mathcal{Q} = \mathcal{Q}_1 \cup \mathcal{Q}_2$ and $0 \in \mathcal{Q}_1$ for all $V$. By \cref{eq:LyapunovParam,eq:LocalLyapunovDecreaseCond} any $V$ satisfies $V(x) \geq 0$ for all $x \in \realsN{n_x}$ and $V(x) - V(x^+) \geq \|x\|^2$ for all $x \in \mathcal{Q}$, respectively. This implies
    \begin{equation}
        \label{eq:Q2Property}
        V(0) < \|x\|^2 \leq V(x), \quad \forall x \in \mathcal{Q}_2.
    \end{equation}
    Next, by construction, $\mathcal{Q}_1$ is compact. Continuity of $V$ in $x$ implies that $\argmin_{x \in \mathcal{Q}_1} \  V(x) = \mathcal{V}_{\textrm{argmin}}$ must be non-empty. In an argument by contradiction, assume $\exists \, {x} \in \mathcal{V}_{\textrm{argmin}}$ with ${x} \neq 0$. Satisfaction of
    \cref{eq:LocalLyapunovDecreaseCond} ensures $V(x^+) \leq V(x) - \|x\|^2$ for all $x \in \mathcal{Q}$, leading to
    \begin{equation}
        \label{eq:MinimumLemmaContradiction}
        V({x}^+)
        \leq V({x}) - \|{x}\|^2 < 
        V({x}) \leq V(0),
    \end{equation}
    where the last inequality follows by assumption. We recall that satisfaction of \cref{eq:BarrierCond} ensures $\mathcal{Q}$ is invariant. Taken together with \cref{eq:Q2Property} the invariance of $\mathcal{Q}$ implies that ${x}^+ \in \mathcal{Q}_1$, which contradicts ${x} \in \mathcal{V}_{\textrm{argmin}}$ thus proving $\mathcal{V}_{\textrm{argmin}} = \{0\}$. By $\argmin_{x \in \mathcal{Q}_1} \  V(x) = \{0\}$ and \cref{eq:Q2Property} it follows $V(0) < V(x)$ for all $x \in \mathcal{Q}\setminus\{0\}$, thereby proving \cref{eq:MinimumOfV}.

    From the above we conclude that $\tilde{V}$ is a positive definite function in $\mathcal{Q}$, which will now be used to construct $\alpha_1$, $\alpha_2 \in \mathcal{K}_\infty$ and continuous, positive definite function $\alpha_3$ from \cref{defn:LyapunovFunction}.
    \begin{enumerate}
            \item By invariance of $\mathcal{Q}$ and positive definiteness of $\tilde{V}$ in $\mathcal{Q}$,
            \begin{equation}
                \begin{alignedat}{2}
                    \tilde{V}(x) &\geq \tilde{V}(x) - \tilde{V}(x^+),\quad  && \forall x \in \mathcal{Q}
                    , \\
                                        &\geq \|x\|^2, && \forall x \in \mathcal{Q}.
                \end{alignedat}
            \end{equation}
            Thus, $\alpha_1(\|x\|) = \|x\|^2$ is a class $\mathcal{K}_{\infty}$ function that lower bounds $\tilde{V}$ for all $x \in \mathcal{Q}$.
            \item Following the work of Kalman \cite{mybibfile:Kalman1960_CT}, by continuity of $\tilde{V}(x)$, 
            \begin{equation}
                \alpha_2(\|x\|) = \sup_{\substack{y \in \mathcal{Q}  \\ \|y\| \leq \|x\|}} \tilde{V}(y) + \|x\|^2,
            \end{equation}
            is a class $\mathcal{K}_{\infty}$ function that upper bounds $\tilde{V}$ for all $x \in \mathcal{Q}$.
            \item By satisfaction of \cref{eq:LocalLyapunovDecreaseCond} it holds, 
            \begin{equation}
                \tilde{V}(x) - \tilde{V}(x^+) \geq \|x\|^2, \quad \forall x \in \mathcal{Q}.
            \end{equation}
            As a result, $\alpha_3(\|x||) = \|x\|^2$ is a continuous, positive definite function upper bounding the difference $\tilde{V}(x) - \tilde{V}(x^+)$ for all $x \in \mathcal{Q}$.
        \end{enumerate}
        Thus, $\tilde{V}$ defined by the solution to augmented optimization problem \cref{eq:OptimProblemFormulationInvariantLocalAsymptoticStability} is a valid Lyapunov function in $\mathcal{Q}$ for closed-loop system \labelcref{eq:GeneralClosedLoopSystem} and under the assumption that $f \circ (\textrm{id}, \varphi)$ is locally bounded, it follows this closed-loop system is \acrshort{LAS} and the set $\mathcal{Q}$ forms part of the \acrshort{ROA} \cite[Thm~B.13]{mybibfile:Rawlings2017}.
\end{proof}

Unfortunately, augmented optimization problem \labelcref{eq:OptimProblemFormulationInvariantLocalAsymptoticStability} is nonconvex, since $\alpha$ and $\sigma_q$ appear bilinearly with specific entries of $\sigma^{\Delta V}_{\textrm{ineq}}$, $\sigma^{\Delta \mathcal{Q}}_{\textrm{ineq}}$, which we denote by
$\sigma^{\Delta V}_{\textrm{ineq}}\big|{}^{}_{\textrm{bi}}$, $\sigma^{\Delta \mathcal{Q}}_{\textrm{ineq}}\big|{}^{}_{\textrm{bi}}$, respectively.
Optimizing over both sets of variables simultaneously would prevent a solution from being obtained efficiently. Therefore, following Valmorbida and Anderson \cite{mybibfile:Valmorbida2014}, augmented optimization problem \labelcref{eq:OptimProblemFormulationInvariantLocalAsymptoticStability} is solved by
\begin{enumerate}
    \item Finding an initial solution via linearization techniques.
    \item Utilizing this initial solution to formulate a sequence of \acrshortpl{SDP} with the property that the \acrshort{ROA} verified by the solution of an \acrshort{SDP} is no smaller than that of the solution from the preceding \acrshort{SDP} in the sequence.
\end{enumerate}
Both steps as well as the overall resulting algorithm are discussed in the remainder of this section. 

\subsubsection{Finding An Initial Solution via Linearization} To formulate the aforementioned sequence of \acrshortpl{SDP}, an initial solution to optimization problem \labelcref{eq:OptimProblemFormulationInvariantLocalAsymptoticStability} is required. Under the assumption that the closed-loop dynamics are continuously differentiable in a neighborhood around the origin, such an initial solution can be found by examining a linearization of the closed-loop dynamics 
\begin{subequations}    
    \begin{align}
        x^+ & = f\big(x, \varphi(x)\big), \\
            & \approx A_{\textrm{lin}}x, \\
            & =
        \Big( \pd{f(x,u)}{x} \Big|_{\begin{subarray}{l} x=0 \\
        u = \varphi(0) \\ \phantom{a} \end{subarray}} + \pd{f(x,u)}{u} \Big|_{\begin{subarray}{l} x=0 \\
        u = \varphi(0)\, \\ \phantom{a} \end{subarray}} \od{\varphi(x)}{x} \Big|_{x=0} \Big)
        x,
    \end{align}
\end{subequations}
where each of the above derivatives can be obtained either analytically or via standard neural network training techniques, e.g. backpropagation. An \acrshort{SDP} can be set up to find a quadratic Lyapunov function, $x\transpose P_{\textrm{lyap}} x$, $P_{\textrm{lyap}} \succ 0$, for this linearized system. If such a Lyapunov function is found, $\mathcal{Q}$ is guaranteed to be an invariant set of the nonlinear system for some $\alpha > 0$ \cite{mybibfile:Khalil1996}. Thus, augmented optimization problem \labelcref{eq:OptimProblemFormulationInvariantLocalAsymptoticStability} may be solved as an \acrshort{SDP} by setting $\sigma_q = x\transpose P_{\textrm{lyap}} x$, $\alpha = \alpha_{\textrm{lyap}}\ll 1$, leading to the initial problem
\begin{subequations}
    \label{eq:SDPFormulationInitialSolution}
    \begin{alignat}{4}
        &\span\span \text{find:} \, &  \; \Pi \setminus \big\{\sigma_q, \, \alpha \big\} \span \nonumber \\
        &\span\span \text{s.t.} \ & \eqref{eq:LyapunovParam}, \, \, \eqref{eq:LocalLyapunovDecreaseCond},  \ &  \eqref{eq:RealizationOfSetQ}\big|\scalebox{0.7}{$ \substack{ 
        \begin{aligned} \sigma_q &= x\transpose P_{\textrm{lyap}} x \\[-0.6ex] \alpha &= \alpha_{\textrm{lyap}} \end{aligned} } $}, \, \, \eqref{eq:BarrierCond} \span  \\
        &\span\span & \begin{gathered} 
        \sigma^V, \sigma^V_{\textrm{ineq}}, \sigma^{\Delta V},  
        \\
        \sigma^{\Delta V}_{\textrm{ineq}}, \sigma^{\Delta \mathcal{Q}},\sigma^{\Delta \mathcal{Q}}_{\textrm{ineq}},
        \end{gathered} \quad & \text{\acrshort{SOS} polynomials,}  \\
        &\span\span & 
        p^{\Delta V}_{\textrm{eq}}, \, p^{\Delta \mathcal{Q}}_{\textrm{eq}} \quad & \text{arbitrary polynomials.} 
    \end{alignat}
\end{subequations}

\subsubsection{A Sequence of SDPs}
Given an initial solution to augmented optimization problem \labelcref{eq:OptimProblemFormulationInvariantLocalAsymptoticStability} by means of the aforementioned linear analysis, we now construct a sequence of \acrshortpl{SDP} whose solutions define a sequence of \acrshort{ROA} estimates $\mathcal{Q}$ that are non-decreasing in size. 

For the first problem in the sequence, assume knowledge of a previous solution to augmented optimization problem \labelcref{eq:OptimProblemFormulationInvariantLocalAsymptoticStability}. Let $\sigma_q^{\mathcal{A}}$ denote the value of $\sigma_q$ in this previous solution. Consider now optimization problem $\mathcal{A}$
\begin{subequations}
    \label{eq:SDPFormulationInvariantLocalAsymptoticStabilityProbA}
    \begin{alignat}{4}
        &\span\span \underset{ 
        \Pi \setminus \{\sigma_q\}
        }{\textrm{maximize:}} \ & & \alpha \span \nonumber \\
        &\span\span \text{s.t.} \ & \eqref{eq:LyapunovParam}, \, \, \eqref{eq:LocalLyapunovDecreaseCond},  \  \eqref{eq:RealizationOfSetQ}\big|\scalebox{0.7}{$\sigma_q = \sigma_q^{\mathcal{A}}$}, \, \, \eqref{eq:BarrierCond} \span  \\
        &\span\span & \begin{gathered} 
        \sigma^V, \sigma^V_{\textrm{ineq}}, \sigma^{\Delta V},  
        \\
        \sigma^{\Delta V}_{\textrm{ineq}}, \sigma^{\Delta \mathcal{Q}},\sigma^{\Delta \mathcal{Q}}_{\textrm{ineq}},
        \end{gathered} \quad & \text{\acrshort{SOS} polynomials,}  \\
        &\span\span & 
        p^{\Delta V}_{\textrm{eq}}, \, p^{\Delta \mathcal{Q}}_{\textrm{eq}} \quad & \text{arbitrary polynomials,} 
    \end{alignat}
\end{subequations}
in which $\sigma_q$ is fixed to the value of $\sigma_q^{\mathcal{A}}$ and is therefore no longer an optimization variable.

Optimization problem $\mathcal{A}$ of \labelcref{eq:SDPFormulationInvariantLocalAsymptoticStabilityProbA} is solved as a sequence of \acrshortpl{SDP} using a line search over $\alpha$ that starts at the value given by the previous solution. The existence of a previous solution at this initial value of $\alpha$ guarantees the feasibility of the first \acrshort{SDP} in this sequence and negates the need for the inclusion of constraint \cref{eq:InvariantLocalAsymptoticStabilityGammaStrictlyPositive}. Additionally, any subsequent solutions to the \acrshortpl{SDP} comprising the line search over $\alpha$ define an \acrshort{ROA} estimate $\mathcal{Q}$ of increased size, as a solution to optimization problem $\mathcal{A}$ of \labelcref{eq:SDPFormulationInvariantLocalAsymptoticStabilityProbA} also defines a solution to augmented optimization problem \cref{eq:OptimProblemFormulationInvariantLocalAsymptoticStability}.

Once an (approximate) solution to optimization problem $\mathcal{A}$ of \labelcref{eq:SDPFormulationInvariantLocalAsymptoticStabilityProbA} has been found, let $\alpha^{\mathcal{A}}$, 
$\mathcal{Q}^{\mathcal{A}}$, 
$\sigma^{\Delta V}_{\textrm{ineq}}\big|{}^{\mathcal{A}}_{\textrm{bi}}$, $\sigma^{\Delta \mathcal{Q}}_{\textrm{ineq}}\big|{}^{\mathcal{A}}_{\textrm{bi}}$
denote the values defined by the relevant variables of this solution. Subsequent \acrshort{SDP} problems in the sequence are defined by solving optimization problem $\mathcal{B}$, in which $\sigma_q$ is an optimization variable. Therefore, to maintain convexity, the values of 
$\sigma^{\Delta V}_{\textrm{ineq}}\big|{}^{}_{\textrm{bi}}$, $\sigma^{\Delta \mathcal{Q}}_{\textrm{ineq}}\big|{}^{}_{\textrm{bi}}$ are fixed to 
$\sigma^{\Delta V}_{\textrm{ineq}}\big|{}^{\mathcal{A}}_{\textrm{bi}}$, $\sigma^{\Delta \mathcal{Q}}_{\textrm{ineq}}\big|{}^{\mathcal{A}}_{\textrm{bi}}$, respectively. In addition, to ensure that any new $\sigma_q$ defines a non-decreasing \acrshort{ROA} estimate, consider inclusion of the constraint
\begin{multline}
    \label{eq:InvariantLocalAsymptoticStabilityQbSupersetConstraint}
    \| \zeta \|^{2k}\big( \sigma_q^{\mathcal{A}}(\zeta) - \sigma_q(\zeta)\big)  \geq p_{\mathcal{B}}(\zeta)\transpose h^\varphi(\zeta)  
    \\ 
    + \sigma_{\mathcal{B}}(\zeta)\transpose
   \begin{bmatrix}
        \mathcal{M}\Big( \left[\!\begin{smallmatrix} \alpha - \sigma_q^{\mathcal{A}}(\zeta) 
        \\
        g^\varphi(\zeta)
        \end{smallmatrix}\!\right], 1 \Big)
        \\
        \mathcal{M}\Big( \left[\!\begin{smallmatrix} \alpha - \sigma_q^{\mathcal{A}}(\zeta) 
        \\
        g^\varphi(\zeta) \end{smallmatrix}\!\right], 2 \Big)
        \\
        \vdots
    \end{bmatrix}, 
    \ \forall \zeta \in \realsN{n_\zeta},
\end{multline}
with $\sigma_{\mathcal{B}}$ a vector of \acrshort{SOS} polynomials, $p_{\mathcal{B}}$ a vector of arbitrary polynomials and $k \in \integersN{}_{\geq 0}$. If satisfied for $\alpha \geq \alpha^{\mathcal{A}}$, \cref{eq:InvariantLocalAsymptoticStabilityQbSupersetConstraint} guarantees 
\begin{equation}
    \sigma_q\big(x, \lambda(x), \varphi(x)\big) \leq \sigma_q^{\mathcal{A}}\big(x, \lambda(x), \varphi(x)\big), \ \forall x \in \bar{\mathcal{Q}}, 
\end{equation}
with $\bar{\mathcal{Q}} = \big\{x \in \realsN{n_x} \mid \alpha - \sigma_q^{\mathcal{A}}(x, \lambda(x), \varphi(x)) \geq 0\big\} \supseteq \mathcal{Q}^{\mathcal{A}}$, thereby guaranteeing $\mathcal{Q} \supseteq \mathcal{Q}^{\mathcal{A}}$. In addition, since $\sigma^{\mathcal{A}}_q(0, \lambda(0), \varphi(0)) = 0$, it follows that satisfaction of \cref{eq:InvariantLocalAsymptoticStabilityQbSupersetConstraint} implies $\sigma_q \in \mathcal{Q}_0$. 
Interpreting \cref{eq:InvariantLocalAsymptoticStabilityQbSupersetConstraint} as a \acrshort{SOS} constraint, optimization problem $\mathcal{B}$ is defined as
\begin{subequations}
    \label{eq:SDPFormulationInvariantLocalAsymptoticStabilityProbB}
    \begin{alignat}{4}
        &\span\span 
        \underset{ \substack{\scalebox{0.65}{$\Pi \setminus \{ \sigma^{\Delta V}_{\textrm{ineq}}|_{\textrm{bi}}, \, \sigma^{\Delta \mathcal{Q}}_{\textrm{ineq}}|_{\textrm{bi}} \}$} \\ \scalebox{0.65}{$\sigma_{\mathcal{B}}, \, p_{\mathcal{B}}$} } }
        {\textrm{maximize:}} \ & \alpha \ & \span \nonumber \\
        &\span\span \text{s.t.} \ & \begin{multlined} \eqref{eq:LyapunovParam}, \, \, \eqref{eq:LocalLyapunovDecreaseCond}\big|{ \scalebox{0.7}{$ \sigma^{\Delta V}_{\textrm{ineq}}|{}^{}_{\textrm{bi}} = \sigma^{\Delta V}_{\textrm{ineq}}|{}^{\mathcal{A}}_{\textrm{bi}}$} }, \quad \\  \quad \quad \ \eqref{eq:RealizationOfSetQ}, \, \, \eqref{eq:BarrierCond}\big|{ \scalebox{0.7}{$ \sigma^{\Delta \mathcal{Q}}_{\textrm{ineq}}|{}^{}_{\textrm{bi}} = \sigma^{\Delta \mathcal{Q}}_{\textrm{ineq}}|{}^{\mathcal{A}}_{\textrm{bi}}$}, \, \, \eqref{eq:InvariantLocalAsymptoticStabilityQbSupersetConstraint} } \end{multlined} \span \span
        \\
        &\span\span & \! \! \! \! \! \! \! \! \! {\scalebox{0.85}{$\begin{gathered} 
        \sigma^V, \sigma^V_{\textrm{ineq}}, \sigma^{\Delta V}, \sigma^{\Delta V}_{\textrm{ineq}},
        \\
        \sigma^{\Delta \mathcal{Q}},\sigma^{\Delta \mathcal{Q}}_{\textrm{ineq}}, \sigma_q, \sigma_{\mathcal{B}}
        \end{gathered}$}} \quad & \text{\acrshort{SOS} polynomials,} 
        \\
        &\span\span & 
        {\scalebox{0.85}{$p^{\Delta V}_{\textrm{eq}}, p^{\Delta \mathcal{Q}}_{\textrm{eq}},p^{\mathcal{B}}$}} \quad & \text{arbitrary polynomials.} 
    \end{alignat}
\end{subequations}
Given the bilinear terms in \cref{eq:InvariantLocalAsymptoticStabilityQbSupersetConstraint}, optimization problem $\mathcal{B}$ is also solved as a sequence of \acrshortpl{SDP} using a line search over $\alpha$ that starts at the value of the previous solution, $\alpha^\mathcal{A}$. At this initial value, optimization problem $\mathcal{B}$ of \labelcref{eq:SDPFormulationInvariantLocalAsymptoticStabilityProbB} is feasible since a solution can be found using the values from optimization problem $\mathcal{A}$ of \labelcref{eq:SDPFormulationInvariantLocalAsymptoticStabilityProbA}, $\sigma_\mathcal{B}=0$ and $p_\mathcal{B} = 0$. Any subsequent solutions to the \acrshortpl{SDP} comprising the line search over $\alpha$ define an \acrshort{ROA} estimate of increased size, as a solution to optimization problem $\mathcal{B}$ of \labelcref{eq:SDPFormulationInvariantLocalAsymptoticStabilityProbB} also defines a solution to augmented optimization problem \cref{eq:OptimProblemFormulationInvariantLocalAsymptoticStability}.

Thus, given an initial solution to optimization problem \labelcref{eq:OptimProblemFormulationInvariantLocalAsymptoticStability} obtained by means of the linear analysis, consider now the sequence of \acrshortpl{SDP} obtained by solving optimization problems $\mathcal{A}$ and $\mathcal{B}$ in an alternating fashion, using a line search method over $\alpha$ as described above. 
It follows that:
\begin{enumerate}
    \item The solution of any \acrshort{SDP} problem in the sequence also defines a solution to augmented optimization problem \labelcref{eq:OptimProblemFormulationInvariantLocalAsymptoticStability}, forming a valid local stability certificate via \cref{thm:ValidityOptimProblemInvariantLocalAsymptoticStability}.
    \item The invariant sets $\mathcal{Q}$ defined by the solutions to the \acrshort{SDP} problems in the sequence are non-decreasing.
\end{enumerate}
Following these observations, a systematic algorithm directly optimizing the \acrshort{ROA} estimate is presented.

\subsubsection{Algorithmic Implementation}
We present a novel algorithm enabling a systematic examination of the local stability properties of closed-loop system \labelcref{eq:GeneralClosedLoopSystem}. This algorithm is shown in pseudocode in \cref{alg:AlternatingLocalStability}.

\begin{algorithm}[t]
    \caption{Sequential Local Stability Analysis}
    \label{alg:AlternatingLocalStability}
    \begin{algorithmic}
        \REQUIRE Semialgebraic set model $\big(\mathbf{K}_\varphi, \mathbf{K}_L\big)$. 
        \ENSURE $\exists P_{\textrm{lyap}} \succ 0$, $\exists
        \alpha_{\textrm{lyap}} 
        \geq \epsilon_{\textrm{prec}}$, $\exists\Delta \alpha_{\,\text{LS}}, \, \alpha_{\textrm{WHILE}}  > 0$ s.t.{\,}:
        \begin{enumerate}
            \renewcommand{\theenumi}{\roman{enumi}}
            \renewcommand{\labelenumi}{(\theenumi)}
            \item $P_{\textrm{lyap}} - A_{\text{lin}}\transpose P_{\textrm{lyap}} A_{\text{lin}} \succ 0$,
            \item \acrshort{SDP} \labelcref{eq:SDPFormulationInitialSolution} is feasible
        \end{enumerate}
        {\setlength{\leftskip}{-0.8em}
        \STATE $(\sigma_q^\mathcal{B}, \Delta \alpha, \alpha^{\mathcal{B}}) \gets (x\transpose P_{\textrm{lyap}} x, \infty, \alpha_{\textrm{lyap}})$
        \WHILE{$\Delta \alpha > \Delta \alpha_{\textrm{WHILE}}$}
        \STATE $(\alpha^{\mathcal{B}}_{\text{old}}, \, \sigma_q^{\mathcal{A}})\gets (\alpha^{\mathcal{B}}, \, \sigma_q^{\mathcal{B}})$
        \STATE $(\alpha^{\mathcal{A}}, \, \sigma^{\Delta V}_{\textrm{ineq}}\big|{}^{\mathcal{A}}_{\textrm{bi}}, \, \sigma^{\Delta \mathcal{Q}}_{\textrm{ineq}}\big|{}^{\mathcal{A}}_{\textrm{bi}}
        ) \gets \text{LineSearch}\big(\alpha^{\mathcal{B}}, \, \alpha^{\mathcal{B}} + \Delta \alpha_{\,\text{LS}},$ 
        \STATE \hfill $\text{Solve}(\text{\acrshort{SDP}} \ \labelcref{eq:SDPFormulationInvariantLocalAsymptoticStabilityProbA}\, ; \sigma_q^{\mathcal{A}})\big)$ 
        \STATE $(\alpha^\mathcal{B}, \, \sigma_q^\mathcal{B}, \, V) \gets \text{LineSearch}\big(\alpha^{\mathcal{A}},  \, \alpha^{\mathcal{A}} + \Delta \alpha_{\,\text{LS}},$ 
        \STATE  \hfill $\text{Solve}(\text{\acrshort{SDP}} \ \labelcref{eq:SDPFormulationInvariantLocalAsymptoticStabilityProbB}\, ; \sigma^{\Delta V}_{\textrm{ineq}}\big|{}^{\mathcal{A}}_{\textrm{bi}}, \, \sigma^{\Delta \mathcal{Q}}_{\textrm{ineq}}\big|{}^{\mathcal{A}}_{\textrm{bi}}, \, \sigma_q^{\mathcal{A}})\big)$ 
        \STATE $\Delta \alpha \gets \alpha^{\mathcal{B}} - \alpha^{\mathcal{B}}_{\text{old}}$
        \ENDWHILE
        \STATE \textbf{Return:} $(V, \alpha^{\mathcal{B}}, \sigma_q^\mathcal{B})$
        \STATE
        }
    \end{algorithmic}
\end{algorithm}
Here, $\text{LineSearch}(\alpha, \, \beta, \text{r})$ denotes a routine applying a line search in the interval $[\alpha, \beta]$ to the routine $\text{r}$, $\text{Solve}\left( \textrm{\acrshort{SDP}} \, (i);\, \gamma\right)$ is a routine determining the feasibility and/or solution of \acrshort{SDP} $(i)$ using the values provided by $\gamma$, $\epsilon_{\textrm{prec}}$ represents a positive scalar value greater than the numerical precision of the \acrshort{SDP} solver, $\Delta \alpha_{\textrm{LS}} > 0$ defines how much $\alpha$ may increase in one iteration of the LineSearch routine, and $\Delta \alpha_{\textrm{WHILE}} > 0$ defines the minimum amount $\alpha$ must increase during one while loop iteration.

From the preceding discussion, it is clear that \Cref{alg:AlternatingLocalStability} addresses the shortcomings identified at the beginning of this section, as it
\begin{enumerate}
    \item systematically updates the set $\mathcal{Q}$ and certifies directly at each \acrshort{SDP} solution that this set forms part of the closed-loop system's \acrshort{ROA},
    \item guarantees that the \acrshort{ROA} estimate $\mathcal{Q}$ does not shrink during each while loop iteration, and
    \item provides a natural termination criteria for the optimization problem in the form of a minimal increase in the value of $\alpha$ after each while loop iteration.
\end{enumerate}
A range of variations of this algorithm may also be implemented, e.g. exiting the while loop after a set number of iterations, etc. See \cref{sec:MPCImitation} for an example.

\section{Numerical Results}
\label{sec:NumericalResults}

The contributions of \cref{sec:ModelingContributions,sec:StabilityVerificationContributions} are demonstrated in the two following numerical examples, respectively. Both examples are solved using the SOSTOOLS \cite{mybibfile:SOSTOOLS} and MOSEK \cite{mybibfile:MOSEK} optimization toolboxes. \acrshortpl{SDP} are solved up to the default numerical tolerance of $\num{1e-6}$. The neural network of \cref{sec:MPCImitation} is trained using the Adam optimizer \cite{mybibfile:Kingma2017}.

\subsection{Implicit Mass-Spring-Damper System}
\label{sec:ImplicitSpringMassDamper}
Consider a mass-spring-damper system with a bounded input and nonlinear, bounded damping force described by 
\begin{equation}
    m\ddot{q} = -kq -d_2\hat{\phi}_{\text{tanh}}(d_1\dot{q}) + \text{sat}(u),
\end{equation}
with saturation achieved at $-1$ and $1$, $m = \qty{1}{\kilogram}$, $k = \qty[scientific-notation=fixed,round-precision=2]{0.5}{\newton\per\meter}$, $c_{\text{tanh}} = \num{1}$, $d_1 = \num[scientific-notation=fixed,round-precision=2]{0.1}$, $d_2 = \qty[scientific-notation=fixed,round-precision=2]{0.5}{\newton\second\per\meter}$. Define the continuous-time state $z\transpose = [q, \dot{q}]\transpose$ and $f$ such that $\dot{z} = f(z,u)$.
Next, a discrete-time system with state $x$ is obtained by approximating the continuous-time solutions via a backward euler discretization at a sampling rate of $T_s = \qty[scientific-notation=fixed,round-precision=1]{0.05}{\second}$, 
\begin{equation}
    x^+ = x + T_s f(x^+, u^+).
\end{equation}
In addition, consider the \acrshort{LQR} state-feedback controller $K = - \vstretch{1.3}{[} \, \num{0.0714}, \ \num{0.742} \vstretch{1.3}{]}$
designed using the approximate system dynamics
\begin{multline}
    x^+ \approx \bigg( I - \dpd{f(x,u)}{x} \bigg|_{\begin{subarray}{l} x=0 \\
    u = 0 \\ \phantom{a} \phantom{a} \end{subarray}} \bigg)^{-1} x \, + \\
    \int_0^{T_s} \bigg( I - \dpd{f(x,u)}{x} \bigg|_{\begin{subarray}{l} x=0 \\
    u = 0 \\ \phantom{a} \end{subarray}} \bigg)^{-1} \big(T_s - \tau\big) \dif \tau \, \bigg( \dpd{f(x,u)}{u} \bigg|_{\begin{subarray}{l} x=0 \\
    u = 0 \\ \phantom{a} \end{subarray}}\bigg) u,
\end{multline}
$Q = \text{diag}(1, 5)$ and $R = 10$. Finally, let $\hat{K} = \vstretch{1.3}{[} \, \hat{k}_1 \ \hat{k}_2 \vstretch{1.3}{]} = 2\sqrt{c_{\text{sp}} + \frac{1}{4}} \, K = - \vstretch{1.3}{[} \, \num{0.122}, \ \num{1.27} \vstretch{1.3}{]}$ and define the control law
\begin{equation}
    u(x) = \hat{\phi}_{\text{sp}}(\hat{K}x+1) - \hat{\phi}_{\text{sp}}(\hat{K}x-1) - 1,
\end{equation}
with $c_{\text{sp}} = \ln{(2)}^2$ such that $u(x) \in [-1, 1]$ for all $x \in \realsN{2}$ and $\od{u(x)}{x} \big|_{\begin{subarray}{l} x = 0_{\phantom{l}} \\  \end{subarray}} = K$. 

This is an implicit system utilizing the semialgebraic activation functions of \cref{sec:ModelingContributions}. Following the developments of \cref{sec:NewNNArchitectures}, the stability properties of this system are analyzed by examining an equivalent closed-loop system consisting of an open-loop dynamical system $\tilde{x}^+ = \tilde{f}(\tilde{x}, \tilde{u}) = \bigl[ \begin{smallmatrix} \tilde{u}_1 \\ \tilde{u}_2 \end{smallmatrix} \bigr]$ with state $\tilde{x} = x$, controlled by a memoryless \acrshort{REN}-based \acrshort{NNC} $\tilde{\varphi}(\tilde{x}) = D_{21}\lambda(\tilde{x}) + D_{22}\tilde{x} + b_{\varphi} = \bigl[ \begin{smallmatrix} 1 & 0 & 0 & 0 & 0 \\ 0 & 1 & 0 & 0 & 0 \end{smallmatrix} \bigr] \lambda(\tilde{x}) + \bigl[\begin{smallmatrix} 0 & 0 \\ 0 & 0 \end{smallmatrix} \bigr] \tilde{x} + \bigl[\begin{smallmatrix} 0  \\  0 \end{smallmatrix} \bigr]$, $\lambda(\tilde{x}) = \phi\big( D_{11} \lambda(\tilde{x}) + D_{12} \tilde{x} + b_v \big)$ where

\begin{subequations}
    \begin{alignat}{5}
        \lambda_{1} &= & \text{id} \Big( & \begin{bNiceArray}{wc{0.46cm}@{\hskip 3pt}wc{0.4cm}@{\hskip 3pt}wc{0.5cm}@{\hskip 4pt}wc{0.4cm}@{\hskip 1pt}wc{0.4cm}@{\hskip 4.5pt}}
            0 & T_s & 0 & 0 & 0
        \end{bNiceArray} 
        \lambda + && \begin{bNiceArray}{wc{0.2cm}@{\hskip 1pt}wc{0.275cm}@{\hskip 3pt}} 1 & 0 \end{bNiceArray} \tilde{x} & \span + 0 && \Big), \\
        \lambda_{2} &= & \text{id} \Big( & \begin{bNiceArray}{wc{0.46cm}@{\hskip 3pt}wc{0.4cm}@{\hskip 3pt}wc{0.5cm}@{\hskip 4pt}wc{0.4cm}@{\hskip 0.5pt}wc{0.4cm}@{\hskip 4.5pt}}
            \frac{-T_s k}{m} & 0 & \frac{-T_s d_2}{m} & \frac{T_s}{m} & \frac{-T_s}{m} 
        \end{bNiceArray} 
        \lambda + && \begin{bNiceArray}{wc{0.2cm}@{\hskip 1pt}wc{0.275cm}@{\hskip 3pt}} 0 & 1 \end{bNiceArray} \tilde{x} & \span - \tfrac{T_s}{m} && \Big), \\
        \lambda_{3} &= & \, \hat{\phi}_{\text{tanh}}\Big( & \begin{bNiceArray}{wc{0.46cm}@{\hskip 3pt}wc{0.4cm}@{\hskip 3pt}wc{0.5cm}@{\hskip 4pt}wc{0.4cm}@{\hskip 1pt}wc{0.4cm}@{\hskip 4.5pt}}
            0 & d_1 & 0 & 0 & 0
        \end{bNiceArray} 
        \lambda +
          && \begin{bNiceArray}{wc{0.2cm}@{\hskip 1pt}wc{0.275cm}@{\hskip 3pt}} 0 & 0 \end{bNiceArray} \tilde{x}
         & \span + 0 && \Big)  , \\
        \lambda_{4} &= & \, \hat{\phi}_{\text{sp}}\Big( & \begin{bNiceArray}{wc{0.46cm}@{\hskip 3pt}wc{0.4cm}@{\hskip 3pt}wc{0.5cm}@{\hskip 4pt}wc{0.4cm}@{\hskip 1pt}wc{0.4cm}@{\hskip 4.5pt}}
            \hat{k}_1 & \hat{k}_2 & 0 & 0 & 0
        \end{bNiceArray} 
        \lambda +
          && \begin{bNiceArray}{wc{0.2cm}@{\hskip 1pt}wc{0.275cm}@{\hskip 3pt}} 0 & 0 \end{bNiceArray} \tilde{x} & \span + 1 && \Big) , \\
        \lambda_{5} &= & \, \hat{\phi}_{\text{sp}}\Big( & \begin{bNiceArray}{wc{0.46cm}@{\hskip 3pt}wc{0.4cm}@{\hskip 3pt}wc{0.5cm}@{\hskip 4pt}wc{0.4cm}@{\hskip 1pt}wc{0.4cm}@{\hskip 4.5pt}}
            \hat{k}_1 & \hat{k}_2 & 0 & 0 & 0
        \end{bNiceArray}
        \lambda +
          && \begin{bNiceArray}{wc{0.2cm}@{\hskip 1pt}wc{0.275cm}@{\hskip 3pt}} 0 & 0 \end{bNiceArray}\tilde{x} & \span
         - 1 && \Big),
    \end{alignat}
\end{subequations}
$b_{x_\varphi} = 0$, $A = B_1 = B_2 = C_1 = C_2 = 0$. The memoryless \acrshort{REN}-based \acrshort{NNC} $\tilde\varphi$ is well-posed since $2I - D_{11} - D_{11}\transpose \succ 0$ and all scalar functions $\textrm{id}$, $\hat{\phi}_{\textrm{tanh}}$ and $\hat{\phi}_{\textrm{sp}}$ are monotone and slope-restricted in $[0,1]$ for $c_{\textrm{tanh}} = 1$, $c_{\textrm{sp}} = \ln(2)^2$.

Application of the semialegbraic set-based modeling procedure to the closed-loop system $\tilde{x}^+ = \tilde{f}\big( \tilde{x}, \tilde{\varphi}(\tilde{x})\big)$ results in an exact semialegbraic system model $\big(\mathbf{K}_{\tilde{\varphi}}, \mathbf{K}_{\tilde{L}}\big)$. Following the procedure of \cref{sec:SoAReview_SDPFormulation}, an \acrshort{SDP} is set up searching for a quartic global Lyapunov function in $x$, $\tilde{V}(x) = \nu(x)\transpose P \nu(x)$ with $\nu(x) =  \vstretch{1.3}{[} 1, \, x_1, \, x_2, \, x_1^2, \, x_1x_2, \, x_2^2 \vstretch{1.3}{]}$. \acrshort{SDP} \labelcref{eq:SDPFormulationGlobalAsymptoticStability} is feasible, verifying that the closed-loop system is \acrshort{GAS}. A phase portrait as well as several level sets of the Lyapunov function with $P =$

\begin{equation*}
        \! 
        \begingroup
            \sisetup{print-exponent-implicit-plus=true}
            \sisetup{tight-spacing=true}
            \sisetup{retain-zero-exponent=false}
            \setlength{\arraycolsep}{2pt}
            \begin{bNiceMatrix}
                0 & \num{9.10e-8}   & -\num{1.31e-7}            & \num{-3.00e-6}            & \num{-5.05e-6}            & \num{-7.59e-7} \\
                * & \num{1.29e3}    & \phantom{-}\num{6.24e2}   & \phantom{-}\num{2.49e1}   & \num{-1.42e1}      & \phantom{-}\num{3.47}e\!+\!0  \\
                * & *               & \phantom{-}\num{1.60e3}   & \phantom{-}\num{1.26e1}   & \phantom{-}\num{2.92e1}   & \num{-8.63}e\!+\!0  \\
                * & *               & *                         & \phantom{-}\num{1.42e2}   & \phantom{-}\num{3.22e-6}  & \phantom{-}\num{3.96e1} \\
                * & *               & *                         & *                         & \phantom{-}\num{4.80e2}   & \phantom{-}\num{2.56e0}e\!+\!0  \\
                * & *               & *                         & *                         & *                         & \phantom{-}\num{5.67e2}
            \end{bNiceMatrix}
        \endgroup
\end{equation*}
are shown in \cref{fig:ImplicitSpringMassDamperLyapunovContour}. 

\begin{figure}[t]
    \begin{center}
    \input{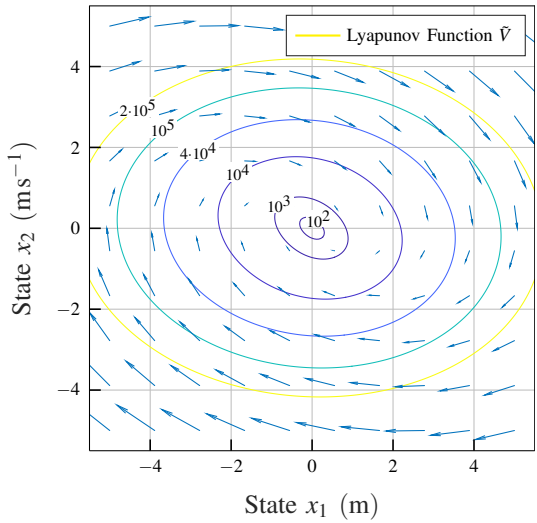}
    \end{center}
    \caption{Phase portrait of the closed-loop system under the neural network control law of \cref{sec:ImplicitSpringMassDamper}, shown in blue. Six level sets of the Lyapunov function $\tilde{V}$ solving \acrshort{SDP} \labelcref{eq:SDPFormulationGlobalAsymptoticStability} are shown.}
    \label{fig:ImplicitSpringMassDamperLyapunovContour}
\end{figure}

\begin{figure}[t]
    \begin{center}
    \input{TubeMPC1qRegion.tex}
    \end{center}
    \caption{Phase portrait of the closed-loop system under the neural-network control law of \cref{sec:MPCImitation}, shown in blue. Overlaid on top are the boundaries of the set $\mathcal{Q}^{\mathcal{A}}$ obtained after iteration $i$ of \cref{alg:AlternatingLocalStability}, and the boundary of the feasible set for the \acrshort{MPC} controller defined by optimization problem \labelcref{eq:LocallyLinearStableSystemTubeBasedMPCOptimizationProblem}.}
    \label{fig:TubeMPC1qRegion}
\end{figure}

\subsection{Model Predictive Controller Imitation}
\label{sec:MPCImitation}
Consider the following saturated, discrete-time \acrshort{LTI} system
\begin{equation}
    \label{eq:LocallyStableLinearSystemOpenLoopDefinition}
    x^+ = Ax + B\text{sat}(u) = \begin{bmatrix} 1 & 0.1 \\ 0 & 1.05 \end{bmatrix} x + \begin{bmatrix} 0 \\ 0.1 \end{bmatrix} \text{sat}(u),   
\end{equation}
with saturation achieved at $1$ and $-1$. 

It is immediately clear that there exists no globally stabilizing controller and therefore a local stability analysis is required as 
\begin{subequations}    
    \begin{alignat}{2}
    	|x_{2}^+| &= |1.05 x_{2} + 0.1 \text{sat}(u)|, \ \ & &   \\
    	 		&\geq 1.05 |x_{2}| - 0.1  & \forall |x_2| \geq & \: \tfrac{0.1}{1.05},
    \end{alignat}
\end{subequations}
which shows that $\{ x \in \realsN{2} \mid |x_2| \geq 2 \}$ is invariant for any control law $u(x)$. Next, consider a \acrshort{NNC} $\varphi$ trained to imitate an \acrshort{MPC} controller,
\begin{equation}
    \label{eq:LocallyLinearStableSystemNeuralNetMPCApproximation}
    x^+ = Ax + B\text{sat}\left(\varphi(x)\right) = Ax + B\left(u_{\text{MPC}}(x) + w\right),
\end{equation}
where $w \in \mathcal{W}$ captures any error between the \acrshort{MPC} controller and a neural-network-based controller. Assume $\mathcal{W} = \{ w \in \realsN{} \mid \|w\|_{\infty} \leq 0.1\}$ and let $u_{\text{MPC}}$ be a tube-based robust \acrshort{MPC} controller such that $u_{\text{MPC}}(x) = K_{\text{tube}}(x - z_1) + v_1$ with $z_1$ and $v_1$ the solutions to
\begin{subequations}
	\label{eq:LocallyLinearStableSystemTubeBasedMPCOptimizationProblem}
	\begin{alignat}{6}
		&& \underset{ \substack{\{z_i\}_{i = \{1, \dotsc, N+1\}} \\ \{v_i\}_{i = \{1, \dotsc, N\} }}}{\textrm{minimize:}} \ && 	\sum_{i=1}^{N} z_i\transpose Q z_i + v_i\transpose R v_i + z_{N+1}\transpose P z_{N+1} \span \span \span \span \span \nonumber \\
		&& \text{s.t.} \span \ & z_{i+1} &= \ && Az_{i}+Bv_{i}, \quad && \forall i\in [N], \label{eq:TubeMPC1NominalDynamicsConstraint}
		\\
		&& 				   && v_{i} &\in \ && \mathcal{U} \ominus K_{\text{tube}} \mathcal{E}, \quad && \forall i\in [N], \label{eq:LocallyLinearStableSystemTubeBasedMPCTightenedInputConstraints}
		\\
		&& 				   && z_{1} &\in \ && x \oplus \mathcal{E}, \quad &&	\\
		&& 				   && z_{N+1} &\in \ && \mathcal{X}_f, \quad &&
	\end{alignat}
\end{subequations}
with $N = 10$, $Q = 5I$, $R = 1$, $\mathcal{U} = [-1, 1]$. In addition, $K_\text{tube}$ and $P$ are equal to the \acrshort{LQR} state-feedback controller and corresponding solution of the discrete-time algebraic Ricatti equation for this choice of $Q$ and $R$, $\mathcal{X}_f$ is the maximal positive invariant set under the aforementioned \acrshort{LQR} controller and tightened input constraints of \cref{eq:LocallyLinearStableSystemTubeBasedMPCTightenedInputConstraints}, and $\mathcal{E}$ is a finite, polytopic approximation of the minimum robust positive invariant set $\oplus_{i=0}^{\infty} (A+BK_{\text{tube}})^iB\mathcal{W}$ \cite{mybibfile:Rakovic2005}. 

A \acrshort{ReLU} neural network consisting of $2$ hidden layers made up of $5$ neurons each is trained in a supervised learning setting using 62500 samples from the above \acrshort{MPC} controller. Next, the weights and biases of the output layer are adjusted to ensure the origin is an equilibrium of the closed-loop system via
\begin{subequations}
	\label{eq:LocallyLinearStableSystemTubeBasedMPCRemoveOffsetOptimizationProblem}
	\begin{alignat}{5}
		&& \underset{ \substack{W^{\textrm{new}}_3, \, b^{\textrm{new}}_3 } }{\textrm{minimize:}} \ && 	\|W^{\textrm{new}}_3 - W_3\|_{\infty} +  \|b^{\textrm{new}}_3 - b_3\|_{\infty} \span \span \span \span \span \nonumber \\
		&& \text{s.t.} \span \ & {W^{\textrm{new}}_3} \lambda_{0} + b^{\textrm{new}}_3 &= 0, &&  &&
	\end{alignat}
\end{subequations}
with $\lambda_0$ equal to the output of the final hidden layer for $x = 0$. Finally, following \cref{eq:LocallyLinearStableSystemNeuralNetMPCApproximation}, two additional \acrshort{ReLU} neurons are added to saturate the controller output between $-1$ and $1$. 

The \acrshort{ROA} of the closed-loop system obtained with this controller is estimated using \cref{alg:AlternatingLocalStability}. An initial solution is found by solving \acrshort{SDP} \cref{eq:SDPFormulationInitialSolution} with the solution to $P_{\textrm{lyap}} - A_{\text{lin}}\transpose P_{\textrm{lyap}} A_{\text{lin}} = I$ and $\alpha_{\textrm{lyap}} = 0.1$. Parameter $\Delta \alpha_{\text{LS}}$ is set to $1$ and the bisection method is used as a line search. All \acrshortpl{SDP} are set up to use a sparse selection of sixth order polynomials. To minimize the size of the positive semidefinite constraint of \cref{eq:InvariantLocalAsymptoticStabilityQbSupersetConstraint}, $\sigma_q$ is defined to be a fourth order \acrshort{SOS} polynomial in $x$, $\sigma_q(x) = \nu(x)\transpose Q \nu(x)$ with $\nu(x) =  \vstretch{1.3}{[} 1, \, x_1, \, x_2, \, x_1^2, \, x_1x_2, \, x_2^2 \vstretch{1.3}{]}$. The termination criteria for each bisection and the algorithm as a whole are set to a relative change in $\alpha$ less than $0.1\%$ or a maximum of $15$ iterations. With these settings \cref{alg:AlternatingLocalStability} terminates after $7$ iterations, with optimization problem $\mathcal{B}$ reported infeasible in the final iteration as a result of the specified numerical tolerances. 

The results prove the closed-loop system is \acrshort{LAS}. In addition, an estimate of the closed-loop system's \acrshort{ROA} is obtained, with the output of \cref{alg:AlternatingLocalStability} being $\sigma_q^{\mathcal{B}}$ defined by $Q=$

\begin{equation*}
        \! 
        \begingroup
            \sisetup{print-exponent-implicit-plus=true}
            \sisetup{tight-spacing=true}
            \sisetup{retain-zero-exponent=false}
            \setlength{\arraycolsep}{1.3pt}
            \begin{bNiceMatrix}
                0 & -\num{5.13e-7}              & -\num{1.56e-6}            & \num{-1.10e-5}                & \num{-3.81e-6}                & \num{-4.03e-6}  \\
                * & \phantom{-}\num{5.88e-3}    & \phantom{-}\num{3.23e-2}  & -\num{1.18e-1}                & \num{-5.72e-2}                & -\num{7.15e-3}  \\
                * & *                           & \phantom{-}\num{2.01e-1}  & -\num{4.81e-1}                & -\num{3.23e-1}                & \num{-4.82e-2}  \\
                * & *                           & *                         & \phantom{-}\num{4.01}e\!+\!0  & \phantom{-}\num{1.94}e\!+\!0  & \phantom{-}\num{4.62e-1} \\
                * & *                           & *                         & *                             & \phantom{-}\num{6.96}e\!+\!0  & \phantom{-}\num{3.24}e\!+\!0  \\
                * & *                           & *                         & *                             & *                             & \phantom{-}\num{3.49}e\!+\!0  \\
            \end{bNiceMatrix}
        \endgroup
\end{equation*}
and $\alpha^\mathcal{B} = 2.84$.

\Cref{fig:TubeMPC1qRegion} contains a phase portrait of the closed-loop system in blue, the boundary of the set $\mathcal{Q}^{\mathcal{A}}$ validated to form part of the closed-loop system's \acrshort{ROA} after every iteration of solving optimization problem $\mathcal{A}$, and the boundary of the feasible set of the tube-based robust \acrshort{MPC} controller used to generate the training data. 

These results show that, as expected, the sets $\mathcal{Q}^{\mathcal{A}}$ are non-decreasing between every iteration of the algorithm. The final set $\mathcal{Q}$ verified to form part of the \acrshort{ROA} is smaller in volume than the \acrshort{MPC} controller's feasible set, but interestingly also proves that points outside of the \acrshort{MPC} controller's feasible set converge to the origin. As \cref{fig:TubeMPC1qRegion} suggests, this behavior is most likely the result of the initial solution used in \cref{alg:AlternatingLocalStability}, which was selected without incorporating prior knowledge of the closed-loop system to demonstrate how \cref{alg:AlternatingLocalStability} can be used without such knowledge.

\section{Conclusion}
\label{sec:Conclusion}
This work presents contributions that address limitations of a state-of-the-art stability verification procedure for \acrshort{NNC}-controlled systems which makes use of semialgebraic-set-based modeling to pose the search for a Lyapunov function as an \acrshort{SDP}. 
Our contributions address the conservatism of the existing verification procedure when used to analyze \acrshortpl{NNC} utilizing transcendental activation functions and the restriction to \acrshortpl{NNC} using feedforward architectures. This is achieved through the introduction of novel semialgebraic activation functions that possess the same fundamental properties of common transcendental activation functions and by demonstrating compatibility of the stability verification procedure with \acrshortpl{NNC} from the broader class of \acrshortpl{REN}, which includes, among others, \acrshortpl{RNN}, \acrshort{LSTM} networks and \acrshortpl{CNN}. Furthermore, the search for a local Lyapunov function is greatly improved via the introduction of a class of candidate Lyapunov functions compatible with the stability verification procedure that is richer than any previously reported and the formulation of a novel sequence of \acrshortpl{SDP} that allow the \acrshort{ROA} estimate to be optimized directly.
Collectively, these contributions enable this stability verification procedure to analyze a broader class of \acrshortpl{NNC} with reduced conservatism.

Future work should aim to further enhance the utility and scalability of the proposed stability verification procedure. Considering the computational challenges associated with solving large-scale \acrshortpl{SDP}, future research should investigate methods that enable larger \acrshortpl{NNC} to be analyzed. One potential approach that we hope to examine in future research is the use of stochastic gradient descent and barrier function techniques to (approximately) solve large scale semidefinite programs \cite{mybibfile:Pauli2022,mybibfile:Revay2020}. If the parameters of the \acrshort{NNC} are also viewed as optimization variables, this method could be used to synthesize \acrshort{NNC} with known (local) stability guarantees. Subsequent research may also focus on examining applications of the newly introduced activation functions to obtain improved semialgebraic-set-based models of existing \acrshortpl{NNC} using transcendental activation functions and using (exact) semialgebraic-set-based models to formulate \acrshortpl{SDP} searching for incremental Lyapunov functions.

\section*{References}
\bibliographystyle{IEEEtranBST/IEEEtran}
\renewcommand{\section}[2]{}%
\bibliography{IEEEtranBST/IEEEabrv,mybibfile}

\begin{thebibliography}{10}
\providecommand{\url}[1]{#1}
\csname url@rmstyle\endcsname
\providecommand{\newblock}{\relax}
\providecommand{\bibinfo}[2]{#2}
\providecommand\BIBentrySTDinterwordspacing{\spaceskip=0pt\relax}
\providecommand\BIBentryALTinterwordstretchfactor{4}
\providecommand\BIBentryALTinterwordspacing{\spaceskip=\fontdimen2\font plus
\BIBentryALTinterwordstretchfactor\fontdimen3\font minus
  \fontdimen4\font\relax}
\providecommand\BIBforeignlanguage[2]{{%
\expandafter\ifx\csname l@#1\endcsname\relax
\typeout{** WARNING: IEEEtran.bst: No hyphenation pattern has been}%
\typeout{** loaded for the language `#1'. Using the pattern for}%
\typeout{** the default language instead.}%
\else
\language=\csname l@#1\endcsname
\fi
#2}}

\bibitem{mybibfile:Hunt1992}
K.~J. Hunt, D.~Sbarbaro, R.~{\.Z}bikowski, and P.~J. Gawthrop,
  ``\BIBforeignlanguage{en}{Neural networks for control systems---a survey},''
  \emph{\BIBforeignlanguage{en}{Automatica (Oxf.)}}, vol.~28, no.~6, pp.
  1083--1112, Nov. 1992.

\bibitem{mybibfile:Hanin2019}
B.~Hanin, ``\BIBforeignlanguage{en}{Universal function approximation by deep
  neural nets with bounded width and {ReLU} activations},''
  \emph{\BIBforeignlanguage{en}{Mathematics}}, vol.~7, no.~10, p. 992, Oct.
  2019.

\bibitem{mybibfile:Gonzalez2024}
\BIBentryALTinterwordspacing
C.~Gonzalez, H.~Asadi, L.~Kooijman, and C.~P. Lim, ``Neural networks for fast
  optimisation in model predictive control: A review,'' arXiv preprint, 2024.
  [Online]. Available: \url{https://arxiv.org/abs/2309.02668}
\BIBentrySTDinterwordspacing

\bibitem{mybibfile:Detailleur2025}
\BIBentryALTinterwordspacing
A.~Detailleur, D.~Wahby, G.~Ducard, and C.~Onder, ``Synthesis and sos-based
  stability verification of a neural-network-based controller for a two-wheeled
  inverted pendulum,'' arXiv preprint, 2025. [Online]. Available:
  \url{https://arxiv.org/abs/2508.15616}
\BIBentrySTDinterwordspacing

\bibitem{mybibfile:Norris2021}
G.~Norris, G.~Ducard, and C.~Onder, ``Neural networks for control: A tutorial
  and survey of stability-analysis methods, properties, and discussions,'' in
  \emph{2021 International Conference on Electrical, Computer, Communications
  and Mechatronics Engineering ({ICECCME})}, Mauritius, Mauritius, Oct. 2021,
  pp. 1--6.

\bibitem{mybibfile:Liberzon2002}
D.~Liberzon, \emph{Switching in systems and control}, ser. Systems \& Control:
  Foundations \& Applications.\hskip 1em plus 0.5em minus 0.4em\relax New York,
  NY: Springer, 2003.

\bibitem{mybibfile:Blondel1999}
V.~D. Blondel and J.~N. Tsitsiklis, ``Complexity of stability and
  controllability of elementary hybrid systems,'' \emph{Automatica (Oxf.)},
  vol.~35, pp. 479--489, Mar. 1999.

\bibitem{mybibfile:Korda2022}
M.~Korda, ``Stability and performance verification of dynamical systems
  controlled by neural networks: Algorithms and complexity,'' \emph{IEEE
  Control Systems Letters}, vol.~6, pp. 3265--3270, June 2022.

\bibitem{mybibfile:Dubach2022}
M.~Dubach and G.~Ducard, ``A comparison of verification methods for
  neural-network controllers using mixed-integer programs,'' in \emph{2022 7th
  International Conference on Robotics and Automation Engineering ({ICRAE})},
  Singapore, Nov. 2022, pp. 43--48.

\bibitem{mybibfile:Schwan2023}
R.~Schwan, C.~N. Jones, and D.~Kuhn, ``Stability verification of neural network
  controllers using mixed-integer programming,'' \emph{IEEE Transactions on
  Automatic Control}, vol.~68, pp. 7514--7529, June 2023.

\bibitem{mybibfile:Hu2020}
H.~Hu, M.~Fazlyab, M.~Morari, and G.~J. Pappas, ``{Reach-SDP}: Reachability
  analysis of closed-loop systems with neural network controllers via
  semidefinite programming,'' in \emph{2020 59th {IEEE} Conference on Decision
  and Control ({CDC})}, Jeju, Korea (South), Dec. 2020, pp. 5929--5934.

\bibitem{mybibfile:Korda2017}
M.~Korda and C.~N. Jones, ``Stability and performance verification of
  optimization-based controllers,'' \emph{Automatica (Oxf.)}, vol.~78, pp.
  34--45, Jan. 2017.

\bibitem{mybibfile:Richardson2023}
C.~R. Richardson, M.~C. Turner, and S.~R. Gunn, ``Strengthened circle and
  {Popov} criteria for the stability analysis of feedback systems with {ReLU}
  neural networks,'' \emph{IEEE Control Systems Letters}, vol.~7, pp.
  2635--2640, June 2023.

\bibitem{mybibfile:Yin2022}
H.~Yin, P.~Seiler, and M.~Arcak, ``Stability analysis using quadratic
  constraints for systems with neural network controllers,'' \emph{IEEE
  Transactions on Automatic Control}, vol.~67, pp. 1980--1987, Apr. 2022.

\bibitem{mybibfile:Pauli2021}
P.~Pauli, D.~Gramlich, J.~Berberich, and F.~Allgower, ``Linear systems with
  neural network nonlinearities: Improved stability analysis via acausal
  zames-falb multipliers,'' in \emph{2021 60th IEEE Conference on Decision and
  Control (CDC)}, Austin, TX, USA, Feb. 2022, pp. 3611--3618.

\bibitem{mybibfile:Newton2022}
M.~Newton and A.~Papachristodoulou, ``Stability of non-linear neural feedback
  loops using sum of squares,'' in \emph{2022 {IEEE} 61st Conference on
  Decision and Control ({CDC})}, Cancun, Mexico, Dec. 2022, pp. 6000--6005.

\bibitem{mybibfile:Revay2024}
M.~Revay, R.~Wang, and I.~R. Manchester, ``Recurrent equilibrium networks:
  Flexible dynamic models with guaranteed stability and robustness,''
  \emph{IEEE Transactions on Automatic Control}, vol.~69, pp. 2855--2870, May
  2024.

\bibitem{mybibfile:Samanipour2024}
P.~Samanipour and H.~A. Poonawala, ``Stability analysis and controller
  synthesis using single-hidden-layer {ReLU} neural networks,'' \emph{IEEE
  Trans. Automat. Contr.}, vol.~69, pp. 202--213, Jan. 2024.

\bibitem{mybibfile:Rawlings2017}
J.~B. Rawlings, D.~Q. Mayne, and M.~Diehl, \emph{Model predictive control:
  Theory, Computation, and Design}.\hskip 1em plus 0.5em minus 0.4em\relax
  Madison, WI: Nob Hill Publishing, 2017.

\bibitem{mybibfile:Newton2021}
M.~Newton and A.~Papachristodoulou, ``Neural network verification using
  polynomial optimisation,'' in \emph{2021 60th IEEE Conference on Decision and
  Control (CDC)}, Austin, TX, USA, Dec. 2021, pp. 5092--5097.

\bibitem{mybibfile:Parrilo2003}
P.~A. Parrilo, ``Semidefinite programming relaxations for semialgebraic
  problems,'' \emph{Mathematical Programming}, vol.~96, pp. 293--320, May 2003.

\bibitem{mybibfile:Revay2020_EquilibriumNet}
\BIBentryALTinterwordspacing
M.~Revay, R.~Wang, and I.~R. Manchester, ``Lipschitz bounded equilibrium
  networks,'' arXiv, 2020. [Online]. Available:
  \url{https://arxiv.org/abs/2010.01732}
\BIBentrySTDinterwordspacing

\bibitem{mybibfile:Agrawal2017}
A.~Agrawal and K.~Sreenath, ``Discrete control barrier functions for
  safety-critical control of discrete systems with application to bipedal robot
  navigation,'' in \emph{Proceedings of Robotics: Science and Systems},
  Cambridge, Massachusetts, July 2017.

\bibitem{mybibfile:Kalman1960_CT}
R.~E. Kalman and J.~E. Bertram, ``Control system analysis and design via the
  'second method' of {Lyapunov},'' \emph{J. Basic Eng.}, vol.~82, pp. 371--393,
  June 1960.

\bibitem{mybibfile:Valmorbida2014}
G.~Valmorbida and J.~Anderson, ``Region of attraction analysis via invariant
  sets,'' in \emph{2014 American Control Conference}, Portland, Oregon, USA,
  June 2014.

\bibitem{mybibfile:Khalil1996}
H.~K. Khalil, \emph{Nonlinear Systems}.\hskip 1em plus 0.5em minus 0.4em\relax
  Upper Saddle River, NJ: Pearson, 1996.

\bibitem{mybibfile:SOSTOOLS}
A.~Papachristodoulou, J.~Anderson, G.~Valmorbida, S.~Prajna, P.~Seiler, P.~A.
  Parrilo, M.~M. Peet, and D.~Jagt, \emph{{SOSTOOLS}: Sum of squares
  optimization toolbox for {MATLAB}}, http://arxiv.org/abs/1310.4716, 2021,
  available from {https://github.com/oxfordcontrol/SOSTOOLS}.

\bibitem{mybibfile:MOSEK}
\emph{MOSEK Optimization Toolbox for MATLAB - Release 10.0.47}, MOSEK ApS,
  MOSEK ApS, Fruebjergvej 3, Symbion Science Park, Box 16, 2100 Copenhagen O,
  Denmark, 2023.

\bibitem{mybibfile:Kingma2017}
\BIBentryALTinterwordspacing
D.~P. Kingma and J.~Ba, ``Adam: A method for stochastic optimization,'' 2017.
  [Online]. Available: \url{https://arxiv.org/abs/1412.6980}
\BIBentrySTDinterwordspacing

\bibitem{mybibfile:Rakovic2005}
S.~V. Rakovic, E.~C. Kerrigan, K.~I. Kouramas, and D.~Q. Mayne, ``Invariant
  approximations of the minimal robust positively invariant set,'' \emph{IEEE
  Transactions on Automatic Control}, vol.~50, pp. 406--410, Mar. 2005.

\bibitem{mybibfile:Pauli2022}
P.~Pauli, N.~Funcke, D.~Gramlich, M.~A. Msalmi, and F.~Allg{\"{o}}wer, ``Neural
  network training under semidefinite constraints,'' \emph{CoRR}, vol.
  abs/2201.00632, 2022.

\bibitem{mybibfile:Revay2020}
M.~Revay, R.~Wang, and I.~R. Manchester, ``A convex parameterization of robust
  recurrent neural networks,'' \emph{IEEE Control Systems Letters}, vol.~5, pp.
  1363--1368, Nov. 2020.

\end{thebibliography}
\appendices
\numberwithin{equation}{section}
\section{\texorpdfstring{Continuity of $(\lambda, \varphi)$ in $\tilde{x}$ for well-posed REN}{Continuity of (lambda, phi) in x tilde for well-posed REN}}
\label{sec:AppendixRENContinuity}
Following a proof similar to that in \cite{mybibfile:Revay2020_EquilibriumNet}, we show how the conditions of well-posedness for a \acrshort{REN} described by \cref{eq:RENdescription} are sufficient to guarantee the continuity of $\lambda$ and $\varphi$ in $\tilde{x} = \bigl[ \begin{smallmatrix} x \\ x_\varphi \end{smallmatrix} \bigr]$. By well-posedness of the \acrshort{REN}, it follows that there exists a $\gamma > 0$, such that
\begin{equation}
\label{eq:RENContinuityInitialCondition}
{\scalebox{0.85}{$
    \begin{aligned}
        2\Lambda & - \Lambda D_{11} -  D_{11}\transpose \Lambda \succ 
        \tfrac{1}{\gamma}\Big( D_{21}\transpose D_{21} + \big(\Lambda \big[ D_{12} \ C_1 \big] + \tfrac{1}{\gamma}D_{21}\transpose \big[ D_{22} \ C_2\big] \big)
        \\ 
        & \big( I - \tfrac{1}{\gamma^2} \Bigl[ \begin{smallmatrix} D_{22}\transpose \\ C_2\transpose \end{smallmatrix} \Bigr] \big[ D_{22} \ C_2 \big] \big)^{-1}
        \big( \Bigl[ \begin{smallmatrix} D_{12}\transpose \\ C_1\transpose \end{smallmatrix} \Bigr] \Lambda + \tfrac{1}{\gamma} \Bigl[ \begin{smallmatrix} D_{22}\transpose \\ C_2\transpose \end{smallmatrix} \Bigr] D_{21} \big)  \Big) \succeq 0,
    \end{aligned}
$}}
\end{equation}
where $\Lambda$ is a diagonal matrix with positive entries. By Schur complement this is equivalent to
\begin{gather}
\label{eq:ContinuityOfRENSchurComplement}
\scalebox{0.85}{$
\begin{bmatrix}
    2\Lambda-\Lambda D_{11} - D_{11}\transpose \Lambda - \frac{1}{\gamma} D_{21}\transpose D_{21} 
    & 
    -\Lambda 
    \big[ D_{12} \ \ C_1 \big] - \tfrac{1}{\gamma}D_{21}\transpose \big[ D_{22} \ \ C_{2} \big]
    \\
    *
    &
    \gamma \Big( I - \frac{1}{\gamma^2} 
    \Bigl[ \begin{smallmatrix} D_{22}\transpose \\ C_2\transpose \end{smallmatrix} \Bigr] \big[ D_{22} \ C_2 \big] \Big)
\end{bmatrix} \succ 0.
$}
\end{gather}
Letting $(\tilde{x}_1, v_1, \lambda_1, \varphi_1)$ and $(\tilde{x}_2, v_2, \lambda_2, \varphi_2)$ with $\tilde{x}_1 \neq \tilde{x}_2$ denote two unique solutions to \cref{eq:RENdescription} and using the shorthand $\Delta_{(\cdot)} = (\cdot)_1 - (\cdot)_2$ it follows
\begin{equation}    
    \label{eq:RENincrementalRelations}
    \begin{alignedat}{2}
         v_1 - v_2 &= \Delta_v && = D_{11}\Delta_{\lambda} + \big[ D_{12} \ \  C_1 \big] \Delta_{\tilde{x}}  \\
         \varphi_1 - \varphi_2 &= \Delta_{\varphi} && = D_{21}\Delta_{\lambda} +  \big[ D_{22} \ \ C_2 \big] \Delta_{\tilde{x}}. 
    \end{alignedat}
\end{equation}
Left and right multiplying \cref{eq:ContinuityOfRENSchurComplement} with $\left[ \begin{smallmatrix} \Delta_{\lambda} \\ \Delta_{\tilde{x}} \end{smallmatrix} \right]$, by \cref{eq:RENincrementalRelations} it follows 
\begin{equation}
    \label{eq:SimplificationQuadraticForm}
    \gamma \Delta_{\tilde{x}}\transpose\Delta_{\tilde{x}} - \tfrac{1}{\gamma}\Delta_\varphi\transpose\Delta_\varphi > 2\langle \Delta_{\lambda},\Delta_{v}-\Delta_{\lambda} \rangle_{\Lambda} \geq 0,
\end{equation}
where the last inequality follows from the assumed monotonicity and slope-restrictedness of all $\phi_i$. From \cref{eq:SimplificationQuadraticForm} it directly follows that $\|\Delta_{\varphi}\| \leq \gamma \|\Delta_{\tilde{x}}\|$, proving (Lipschitz) continuity of $\varphi$ in $\tilde{x}$. Lipschitz continuity of $v$ in $\tilde{x}$ is proven using a similar condition to \cref{eq:RENContinuityInitialCondition}. Under the assumption that all $\phi_i$ are continuous, it follows that $\lambda$ is also continuous in $\tilde{x}$.

\vskip -2\baselineskip plus -1fil
\begin{IEEEbiographynophoto}{Alvaro Detailleur}
received the master's degree in Robotics, Systems and Control from ETH Z{\"u}rich, Z{\"u}rich, Switzerland in 2024, focusing on system modeling and model-based control design. 

He has completed industrial internships at Forze Hydrogen Racing and Mercedes-AMG High Performance Powertrains, focusing on the software and control of high performance automotive power units. His research interests include nonlinear control, neural-network-based controllers and optimization-based control design with a practical application.

Mr. Detailleur was awarded the Young Talent Development Prize by the Royal Dutch Academy of Sciences (KHMW).
\end{IEEEbiographynophoto}

\vskip -2\baselineskip plus -1fil
\begin{IEEEbiographynophoto}{Dalim Wahby} received the B.Sc. degree in industrial engineering and management from Karlsruhe Institute of Technology (KIT), Karlsruhe, Germany, in 2022, with a focus on energy technologies and natural language processing. Additionally, he received the Dipl. Ing. in electronics and embedded systems from Polytech Nice-Sophia, Sophia Antipolis, France in 2024, and the M.Sc. in ICT innovation from Royal Institute of Technology (KTH), Stockholm, Sweden in 2025, with a major in electrical engineering.

He has completed a research internship at CNRS, focusing on adaptive control and the stability analysis of neural-network-based controllers. Currently, he is pursuing the Ph.D. degree in automatic signal and image processing at i3S/CNRS in Sophia-Antipolis, under the supervision of Guillaume Ducard, focusing on the development of a framework for the design and the analysis of neural-network-based controllers.
\end{IEEEbiographynophoto}

\vskip -2\baselineskip plus -1fil
\begin{IEEEbiographynophoto}{Guillaume J. J. Ducard} (Senior Member, IEEE), received the master’s degree in electrical engineering and the
Doctoral degree focusing on flight control for unmanned aerial vehicles (UAVs) from ETH
Z{\"u}rich, Z{\"u}rich, Switzerland, in 2004 and 2007, respectively.

He completed his two-year Postdoctoral course in 2009 from ETH Z{\"u}rich, focused
on flight control for UAVs. He is currently an Associate Professor with the
Universit{\'e} C\^{o}te d`Azur, France, and guest scientist with ETH Z{\"u}rich. His research interests include nonlinear control, neural networks, estimation, and guidance mostly applied to UAVs.
\end{IEEEbiographynophoto}
\vskip -2\baselineskip plus -1fil
\begin{IEEEbiographynophoto}{Christopher H. Onder} received the Diploma in mechanical engineering and Doctoral
degree in Doctor of technical sciences from ETH Zürich, Zürich, Switzerland.

He is a Professor with the Institute for Dynamic Systems and Control, Department of
Mechanical Engineering and Process Control, ETH Zürich. He heads the Engine Systems
Laboratory, and has authored and co-authored numerous articles and a book on modeling and
control of engine systems. His research interests include engine systems modeling, control
and optimization with an emphasis on experimental validation, and industrial cooperation.

Prof. Dr. Onder was the recipient of the BMW scientific award, the ETH medal, the
Vincent Bendix award, the Crompton Lanchester Medal, and the Arch T. Colwell award.
Additionally, he was awarded the Watt d’Or, the energy efficiency price of the
Swiss Federal Office of Energy, on multiple occasions for his projects. 
\end{IEEEbiographynophoto}
\vfill
\end{document}